\documentclass[a4paper,11pt]{article} 
\usepackage[margin=1in]{geometry}

\usepackage{graphicx}
\usepackage{amsthm}
\makeatletter
\def\th@plain{%
	\thm@notefont{}% same as heading font
	\itshape % body font
}
\def\th@definition{%
	\thm@notefont{}% same as heading font
	\normalfont % body font
}
\makeatother

\usepackage{cite}
\usepackage{amssymb}
\usepackage{amsmath}
\usepackage{mathrsfs}
\usepackage{verbatim}
\usepackage{algorithm}
\usepackage{mathtools}
\usepackage{xspace}
\usepackage{color}
\usepackage{url}
\usepackage[all]{xy}
\usepackage{microtype}
\usepackage{array}
\usepackage{float}
\usepackage{algorithmic}
\usepackage{lmodern}
\usepackage{anyfontsize}
\usepackage{hyperref}
\usepackage{cleveref}
\usepackage[normalem]{ulem}
\usepackage{epigraph}
\usepackage{wasysym}
\usepackage{dsfont}
\usepackage{tikz}
\usepackage{enumerate}
\usepackage{thmtools}
\usepackage{thm-restate}

\pdfoutput=1

\newcommand{\R}{\mathbb{R}}

\newcommand{\orig}{\mathbf{0}}
\newcommand{\av}{\mathbb{E}}

\newcommand{\qgraph}{\mathcal{G}}
\newcommand{\tree}{\mathcal{T}}

\newcommand{\diam}{\mathrm{diam}}

\newcommand{\dotp}[2]{\langle#1,#2\rangle}

\newcommand{\ignore}[1]{}

\newcommand{\conv}{\mathrm{conv}}

\newcommand{\cp}{\mathrm{cp}}

\newcommand{\Caratheodory}{Carath\'eodory\xspace}
\newcommand{\Barany}{B\'ar\'any\xspace}
\newcommand{\Soberon}{Sober\'{o}n\xspace}
\newcommand{\maxdegree}{\Delta(\qgraph)}
\newcommand{\gdiam}{\mathrm{diam}(\qgraph)}
\newcommand{\numedges}{\|\qgraph\|}

\DeclarePairedDelimiter{\ceil}{\lceil}{\rceil}
\DeclarePairedDelimiter{\floor}{\lfloor}{\rfloor}

\newtheorem{lemma}{Lemma}
\numberwithin{lemma}{section}

\newtheorem{theorem}[lemma]{Theorem}

\newtheorem{remark}[lemma]{Remark}

\begin{document}

\title{No-dimensional Tverberg Theorems and Algorithms\thanks{Supported in
part by ERC StG 757609. A preliminary version appeared as
A.~Choudhary and W.~Mulzer. \emph{No-dimensional Tverberg theorems and 
algorithms}, Proc.~36th SoCG, pp.~31:1--31:17, 2020.
}
}

\author{
Aruni Choudhary\thanks{Institut f\"ur Informatik, 
Freie Universit\"at Berlin. \texttt{[arunich, mulzer]@inf.fu-berlin.de}.
}
\and
Wolfgang Mulzer\footnotemark[2]
}

\maketitle

\begin{abstract}
Tverberg's theorem 
states that for any 
$k \ge 2$ and any 
set $P \subset \R^d$ of
at least $(d + 1)(k - 1) + 1$ points in $d$ dimensions,
we can partition $P$ into $k$ subsets 
whose convex hulls have a non-empty 
intersection. The associated search problem of finding the partition
lies in the complexity 
class $\text{CLS} = \text{PPAD} \cap \text{PLS}$, but
no hardness results are known.
In the \emph{colorful} Tverberg theorem,
the points in $P$ have colors,
and under certain conditions, 
$P$ can be partitioned into  \emph{colorful} 
sets, in which each color appears
exactly once and whose convex hulls 
intersect. To date, the complexity 
of the associated search problem 
is unresolved.
Recently,  Adiprasito, \Barany, and 
Mustafa~[SODA 2019] gave a \emph{no-dimensional}
Tverberg theorem, in which the
convex hulls  
may intersect in an \emph{approximate}
fashion. This relaxes the
requirement on the cardinality 
of $P$. 
The argument is constructive, 
but does not result in a 
polynomial-time algorithm.

We present a deterministic algorithm that
finds for any $n$-point set $P \subset \R^d$ and 
any $k \in \{2, \dots, n\}$
in $O(nd \ceil{\log k})$ time
a $k$-partition of $P$ such that 
there is a ball of radius $O\left((k/\sqrt{n})\mathrm{diam(P)}\right)$ 
that intersects the convex hull of each set.
Given that 
this problem is not known to be solvable 
exactly in 
polynomial time, our result provides 
a remarkably efficient and simple new notion of 
approximation.

Our main contribution is to generalize 
Sarkaria's method~[Israel Journal Math., 1992] 
to reduce the Tverberg problem to the 
Colorful \Caratheodory problem (in the simplified tensor 
product interpretation of \Barany and Onn) and to apply it algorithmically.
It turns out that this not only leads to an alternative algorithmic 
proof of a no-dimensional Tverberg theorem, but it also generalizes 
to other settings such as the colorful variant of the problem.
\end{abstract}

\section{Introduction}
\label{section:introduction}

In 1921, Radon~\cite{radon} proved a seminal 
theorem in convex geometry:
given a set $P$ of at least $d + 2$ points in 
$\R^d$, one can always split
$P$ into two  non-empty sets whose convex hulls 
intersect. In 1966, 
Tverberg~\cite{tverberg-original} generalized 
Radon's theorem to allow for more sets in 
the partition. Specifically, he showed that 
for any $k \geq 1$, if a
$d$-dimensional point set $P\subset\R^d$ has cardinality 
at least $(d + 1)(k - 1) + 1$, then $P$ can be partitioned 
into $k$ non-empty, pairwise disjoint sets $T_1, \dots, T_k \subset P$ 
whose convex hulls have a non-empty 
intersection, i.e., 
$\bigcap_{i = 1}^k \conv(T_i)  \neq \emptyset$, 
where $\conv(\cdot)$ denotes the convex hull.

By now, several alternative proofs of Tverberg's 
theorem are known, 
e.g.,~\cite{tverberg-second,tv-generalization,roudneff-conic,MillerSh10,sarkaria,bo-tverberg,abbfm-colorful,BaranyBlZi16}.
Perhaps the most elegant proof is due to 
Sarkaria~\cite{sarkaria}, with 
simplifications by \Barany and Onn~\cite{bo-tverberg} 
and by Aroch et al.~\cite{abbfm-colorful}.
In this paper, all further references to \emph{Sarkaria's method} refer
to the simplified version.
This proof proceeds by a reduction to the 
\emph{Colorful \Caratheodory theorem}, another 
celebrated result in convex geometry:
given $r \geq d + 1$ 
point sets $P_1, \dots, P_r \subset \R^d$ 
that have a common point $y$ in their 
convex hulls $\conv(P_1), \dots, \conv(P_r)$,
there is a \emph{traversal}
$x_1 \in P_1, \dots, x_r \in P_r$,
such that $\conv(\{x_1, \dots, x_r\})$
contains $y$. 
A two-dimensional example is given in Figure~\ref{figure:colca-basic}.
\begin{figure}
\centering
\includegraphics[width=0.43\textwidth,page=2]{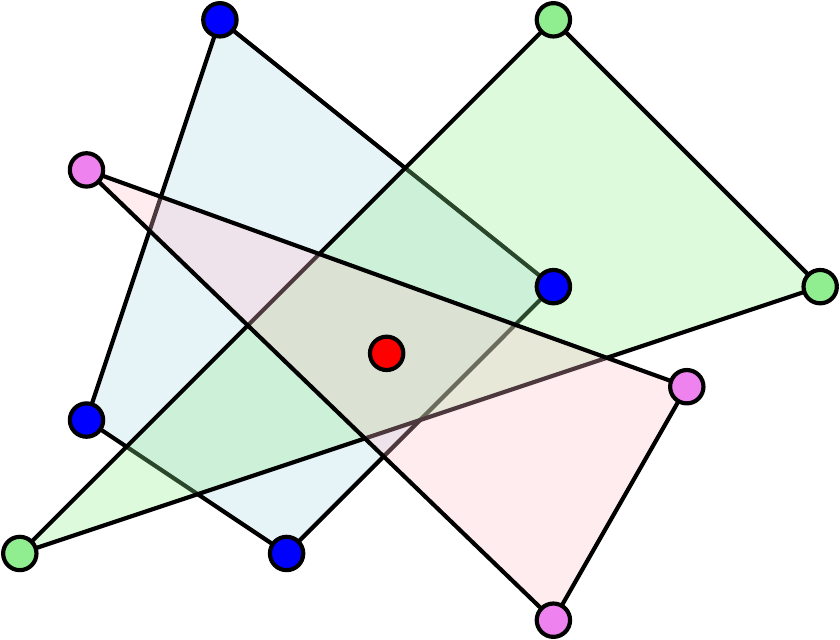}\hspace{1em}
\includegraphics[width=0.43\textwidth,page=4]{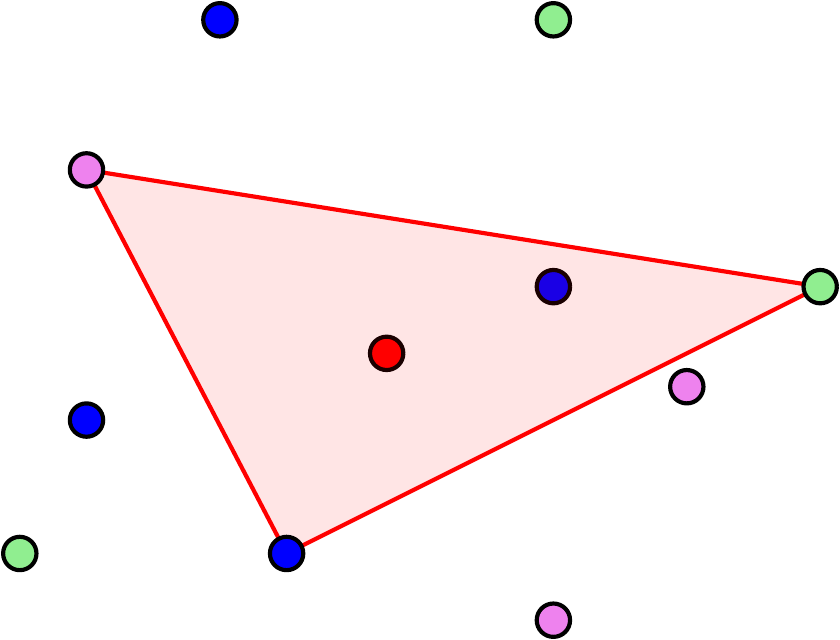}
\caption{The Colorful \Caratheodory theorem. Left: the convex 
hulls of the three point sets intersect; 
Right: a colorful triangle that contains the common point.}
\label{figure:colca-basic}
\end{figure}
Sarkaria's proof~\cite{sarkaria} 
uses a tensor product to lift the original points
of the Tverberg instance into higher
dimensions, and then 
uses the Colorful \Caratheodory
traversal to obtain a Tverberg partition 
for the original point set.

From a computational point of view, a Radon 
partition is easy to find by 
solving $d + 1$ linear equations.
On the other hand, finding Tverberg partitions 
is not straightforward.  Since a Tverberg partition 
must exist if 
$P$ is large enough, finding such a partition 
is a total search problem. In fact, the problem of computing 
a Colorful \Caratheodory traversal
lies in the complexity class 
$\text{CLS} = \text{PPAD} \cap \text{PLS}$~\cite{MulzerSt18,mmss-cc},
but no better upper bound 
is known.
Sarkaria's proof gives a polynomial-time reduction from
the problem of finding a Tverberg partition to the 
problem of finding a colorful traversal, thereby placing the 
former problem in the same complexity class.
Again, as of now we do not know better upper bounds
for the general problem.
Miller and Sheehy~\cite{MillerSh10} and 
Mulzer and Werner~\cite{mw-approx} provided algorithms
for finding \emph{approximate} Tverberg partitions,
computing a partition into fewer sets than is guaranteed
by Tverberg's theorem in time 
that is linear in $n$, but 
quasi-polynomial in the dimension.
These algorithms were motivated by applications in mesh 
generation and statistics that require finding a point
that lies ``deep'' in $P$. A point in the common intersection of 
the convex hulls of a Tverberg partition has this property,
with the partition serving as a 
certificate of depth.
Recently Har-Peled and Zhou have proposed algorithms~\cite{hpz-tverberg}
to compute approximate Tverberg partitions that take time 
polynomial in $n$ and $d$.

Tverberg's theorem also admits a colorful variant, 
first conjectured 
by \Barany and Larman~\cite{bl-colorful}.
The setup consists of $d + 1$ point sets 
$P_1, \dots, P_{d+1} \subset \R^d$,
each set interpreted as a different color and having size $t$.
For a given $k$, the goal is to find $k$ pairwise-disjoint 
\emph{colorful} sets (i.e., each set contains
at most one point from each $P_i$) $A_1,\dots,A_k$ such that 
$\bigcap_{i=1}^{k} \conv(A_i) \neq \emptyset$.
The problem is to determine the optimal value of $t$ 
for which such a colorful partition always exists.
\Barany and Larman~\cite{bl-colorful} conjectured that $t=k$ suffices
and they proved the conjecture 
for $d = 2$ and arbitrary $k$, and
for $k =2 $ and arbitrary $d$.
The first result for the general case was given by 
\v{Z}ivaljevi\'c and Vre\'cica~\cite{zv-colorful} 
through topological arguments.
Using another topological argument,
Blagojevi\v{c}, Matschke, and Ziegler~\cite{bmz-colorful} showed that
(i) if $k + 1$ is prime, then $t = k$;
and (ii) if $k + 1$ is not prime, then  
$k\le t \le 2k-2$.
These are the best known bounds for arbitrary $k$.
Later Matou\v{s}ek, Tancer, and Wagner~\cite{mtw-geometric} gave 
a geometric proof that is inspired
by the proof of Blagojevi\v{c}, Matschke, and Ziegler~\cite{bmz-colorful}.

More recently, Sober\'on~\cite{soberon-equal} showed 
that if more color classes are available,
then the conjecture holds for any $k$.
More precisely, for  $P_1,\dots,P_n \subset \R^d$ with 
$n = (k - 1)d + 1$, each of size $k$,
there exist $k$ colorful sets whose convex hulls intersect.
Moreover, there is a point in the common intersection
so that the coefficients of its convex combination 
are the same for each colorful set in the partition.
The proof uses Sarkaria's tensor product construction.

Recently Adiprasito, \Barany, and 
Mustafa~\cite{abm-nodim} established 
a relaxed version of the Colorful 
\Caratheodory theorem and some of its descendants~\cite{barany-cc}.
For the Colorful \Caratheodory theorem, this 
allows for a (relaxed) traversal of arbitrary size, 
with a guarantee that the convex hull of the traversal is close to the common point $y$.
For the Colorful Tverberg problem, they prove a version of the 
conjecture where the convex hulls of the colorful sets intersect 
approximately.
This also gives a relaxation for 
Tverberg's theorem~\cite{tverberg-original} that
allows arbitrary-sized partitions, again with an approximate
notion of intersection.
Adiprasito et al.~refer to these results as 
\emph{no-dimensional} versions 
of the respective classic theorems, because
the dependence on the
ambient dimension is relaxed.
The proofs use averaging arguments.
The argument for the no-dimensional Colorful \Caratheodory theorem
also gives an efficient algorithm to find a suitable traversal.
However, the arguments for the no-dimensional Tverberg theorem results do 
not give a polynomial-time algorithm for finding the 
partitions.

\paragraph{Our contributions.}
We prove no-dimensional variants of the Tverberg theorem 
and its colorful counterpart that allow for efficient algorithms.
Our proofs are inspired by Sarkaria's method~\cite{sarkaria} and the 
averaging technique by Adiprasito, \Barany, and Mustafa~\cite{abm-nodim}.
For the colorful version, we additionally make use of ideas of
\Soberon~\cite{soberon-equal}. Furthermore, we also 
give a no-dimensional generalized Ham-Sandwich theorem~\cite{zv-extension}
that interpolates  between the Centerpoint theorem and the Ham-Sandwich 
theorem~\cite{st-sandwich}, again with an efficient algorithm. 

Algorithmically, Tverberg's theorem is useful for finding 
centerpoints of high-dimensional point sets, which in turn 
has applications in statistics and mesh generation~\cite{MillerSh10}.
In fact, most algorithms for finding centerpoints are 
Monte-Carlo, returning some point $p$ and a probabilistic guarantee 
that $p$ is indeed a centerpoint~\cite{ClarksonEMST96,Har-PeledJ19}. 
However, this is coNP-hard to verify.
 On the other hand, a (possibly approximate) 
Tverberg partition immediately gives a certificate of 
depth~\cite{MillerSh10,mw-approx}.
Unfortunately, there are no polynomial-time algorithms for finding
optimal Tverberg partitions.
In this context, our result provides a fresh notion of approximation 
that also leads to very fast polynomial-time algorithms. 

Furthermore, the Tverberg problem is intriguing from a
complexity theoretic point of view, because it constitutes 
a total search problem that is not known to be solvable 
in polynomial time, but which is also unlikely to be NP-hard.
So far, such problems have mostly been studied in the 
context of algorithmic game theory~\cite{NRTV2007}, and only 
very recently a similar line of investigation has been 
launched for problems in high-dimensional discrete 
geometry~\cite{mmss-cc,MulzerSt18,Filos-RatsikasG19, dgmm-survey}.
Thus, we show that the \emph{no-dimensional} variant of 
Tverberg's theorem is easy from this point of view.
Our main results are as follows:

\begin{itemize}
\item 
Sarkaria's method uses a specific set of $k$ vectors in $\R^{k-1}$ 
to lift the points in the Tverberg instance to 
a Colorful \Caratheodory instance.
We refine this method to vectors that are defined with
the help of a given graph.
The choice of this graph is important in proving good 
bounds for the partition and in the algorithm.
We believe that this generalization is of 
independent interest and may prove useful
in other scenarios that rely on the tensor product construction.

\item Let $\diam(x)$ denote the diameter of any set $x$.
We prove an efficient no-dimensional Tverberg result:
\begin{restatable}[efficient no-dimensional Tverberg]{theorem}{ndtverberg}
\label{theorem:tverberg}
Let $P$ be a set of $n$ points in $d$ dimensions, 
and let $k \in \{2, \dots, n\}$ be an integer.
\begin{enumerate}[(i)]
\item \label{tverberg1} 
For any choice of 
positive integers $r_1,\dots,r_k$ that satisfy $\sum_{i=1}^{k} r_i = n$,
there is a partition $T_1, \dots, T_k$ of $P$ with 
$|T_1| = r_1, |T_2|=r_2, \dots , |T_k|=r_k$,
and a ball $B$ of radius 
\[
\frac{n\,\diam(P)}{\min_i r_i} \sqrt{\frac{10 \ceil{\log_4 k} }{n-1}}
=O\left(\frac{\sqrt{n\log k}}{\min_i r_i} \,\diam(P)\right)
\]
such that $B$ intersects the convex hull of each $T_i$.

\item \label{tverberg2}  
The bound is better for the case $n=rk$ and $r_1=\dots=r_k=r$.
There exists a partition $T_1, \dots, T_k$ of $P$ with 
$|T_1| = \dots = |T_k| = r$ and a $d$-dimensional ball of radius
\[
\sqrt{\frac{k(k-1)}{n-1}}\diam(P)
=O\left(\frac{k}{\sqrt{n}}\diam(P)\right)
\]
that intersects the convex hull of each $T_i$.

\item \label{tverberg3}  
In either case, the partition $T_1, \dots, T_k$
can be computed in deterministic time
\[
O(nd\ceil{\log k}).
\]
\end{enumerate}

\end{restatable}
See Figure~\ref{figure:nodim-partition} for a simple illustration.

\begin{figure}
\centering
\includegraphics[width=0.43\textwidth,page=5]{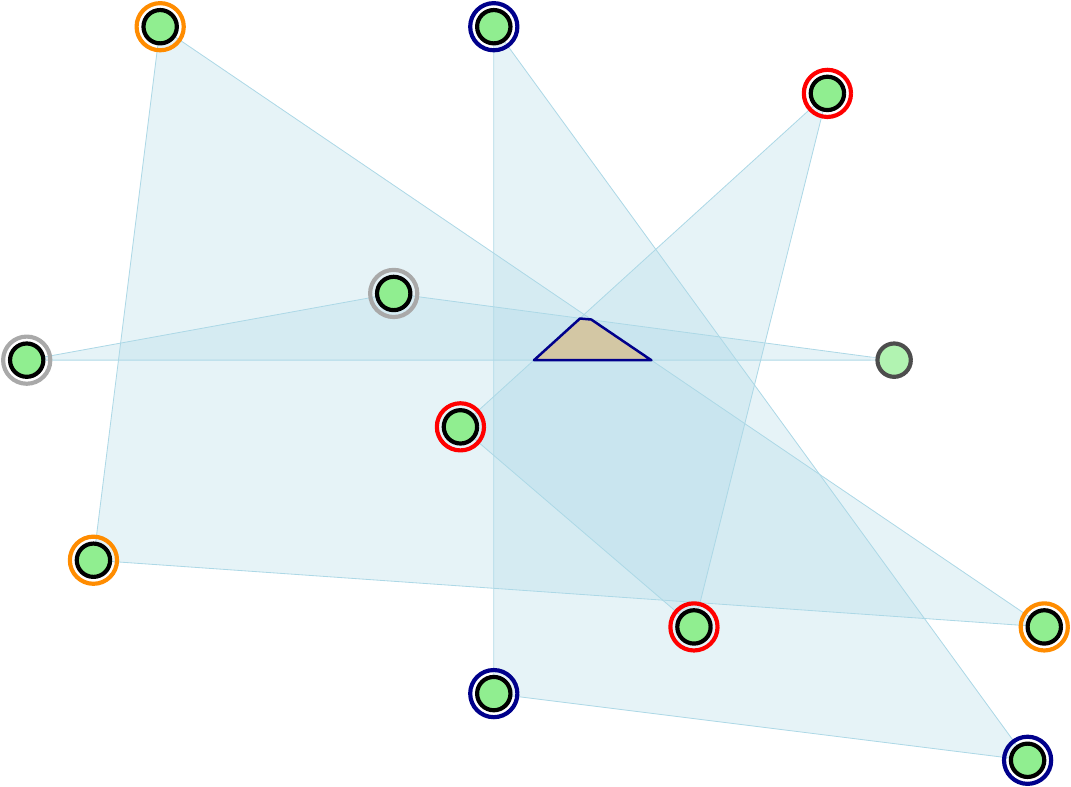}\hspace{1em}
\includegraphics[width=0.43\textwidth,page=3]{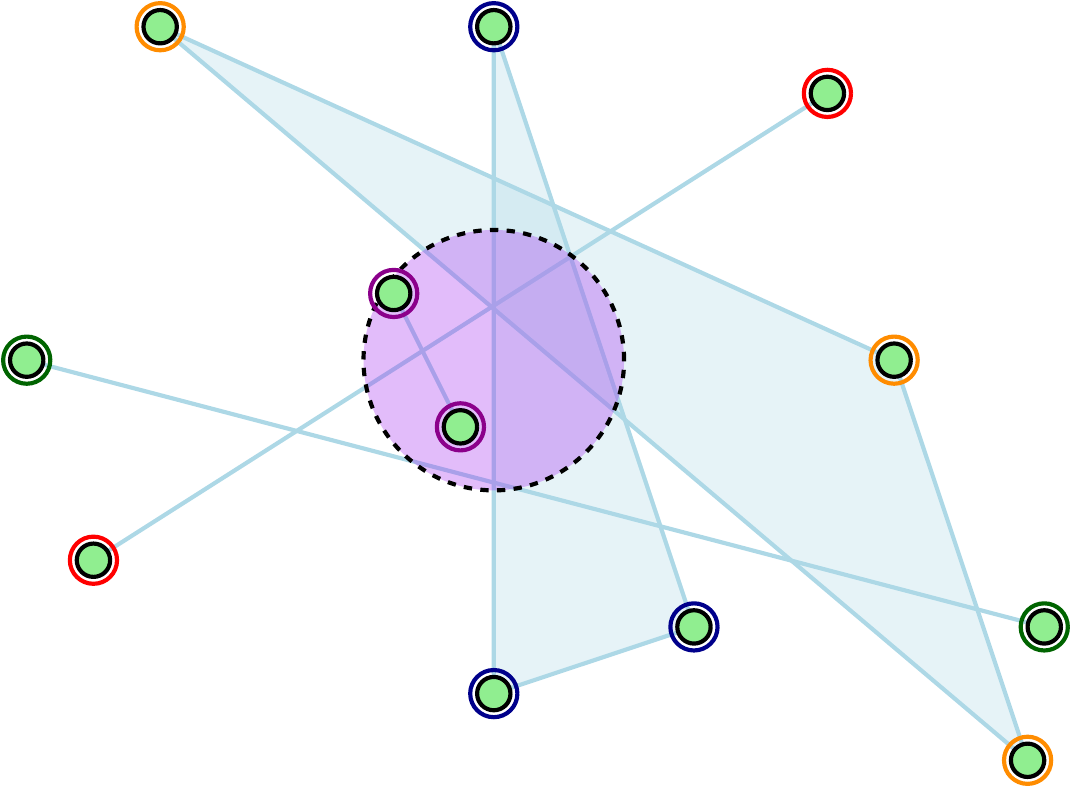}
\caption{Left: a 4-partition of a planar point set. Larger 
Tverberg partitions are not possible
because there are not enough points.
Right: a 5-partition on the same point set with a disk intersecting the convex hulls of 
each set of the partition.}
\label{figure:nodim-partition}
\end{figure}

\item 
and a colorful counterpart (for a simple example, see 
Figure~\ref{figure:colored-partition}):
\begin{restatable}[efficient no-dimensional Colorful Tverberg]{theorem}{ndcolorfultverberg}
\label{theorem:colorful-tverberg}
Let $P_1$, $\ldots$, $P_n\subset \R^d$ be point sets, each of size $k$,
with $k$ being a positive integer, 
so that the total number of points is $N = nk$.
\begin{enumerate}[(i)]
\item \label{ctverberg1}
Then, there are $k$ pairwise-disjoint colorful sets $A_1,\dots,A_k$ 
and a ball of radius
\[
\sqrt{\frac{2k(k-1)}{N}}\max_i\diam(P_i)
=O\left(\frac{k}{\sqrt{N}}\,\max_i\diam(P_i)\right)
\]
that intersects $\conv(A_i)$ for each $i\in [k]$.

\item \label{ctverberg2}
The colorful sets $A_1,\dots,A_k$ can be computed in deterministic time
$
O(Ndk).
$
\end{enumerate}
\end{restatable}

\item 
For any sets $P, x\subset \R^d$, the \emph{depth} of 
$x$ with respect to $P$ is the largest positive integer $k$ 
such that every half-space that contains $x$ also contains 
at least $k$ points of $P$.
\begin{restatable}[no-dimensional Generalized Ham-Sandwich]{theorem}{ndinterpolate}
\label{theorem:nodim-gghs}
Let $k$ finite point sets $P_1$, $\ldots$, $P_k$ in $\R^d$ be given, 
and let $m_1, \dots, m_k$, $2 \leq m_i \leq |P_i|$ for $i \in [k]$, $k\le d$,
be any set of integers.

\begin{enumerate}[(i)]
\item \label{gghs1}
There is a linear transformation and a 
ball $B\in \R^{d-k+1}$ of radius 
\[
(2 + 2\sqrt{2}) \max_i \frac{\diam(P_i)}{\sqrt{m_i}},
\]
such that the hypercylinder $B\times \R^{k-1}\subset \R^d$
has depth at least $\left\lceil |P_i|/m_i \right\rceil$
with respect to $P_i$, for $i \in [k]$, after applying the transformation.

\item \label{gghs2} 
The ball and the transformation can be determined in time
\[
O\left(d^6+dk^2+\sum_{i} |P_i|d\right).
\]
\end{enumerate}
\end{restatable}
\end{itemize}

\begin{figure}
\centering
\includegraphics[width=0.43\textwidth,page=1]{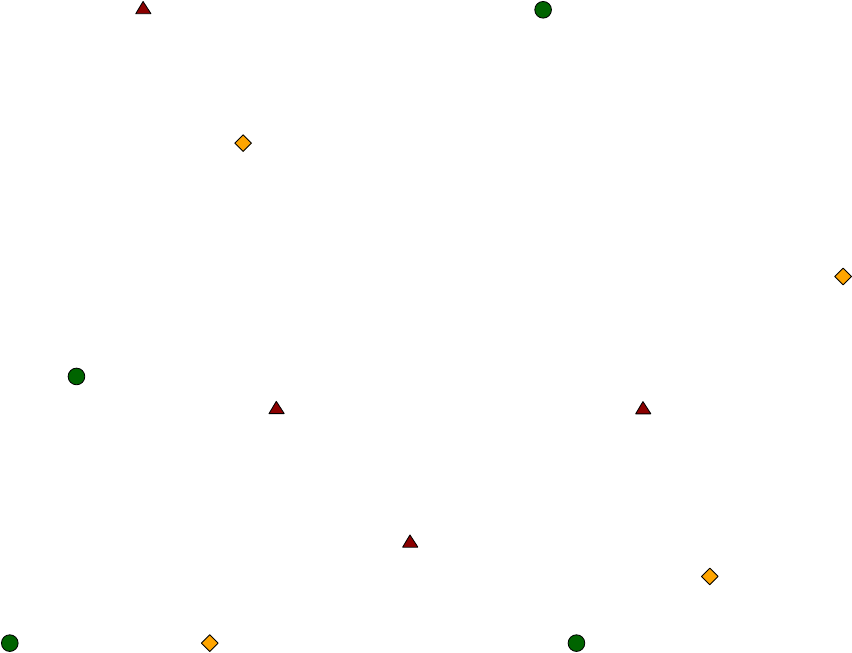}\hspace{1em}
\includegraphics[width=0.43\textwidth,page=2]{colorful_tverberg}
\caption{Left: a point set on three colors and four points of each color.
Right: a colorful partition with a ball containing
the centroids (squares) of the sets of the partition.}
\label{figure:colored-partition}
\end{figure}

The colorful Tverberg result is similar in spirit to 
the regular version,
but from a computational viewpoint, it 
does not make sense to use the colorful algorithm to 
solve the regular Tverberg problem.

Compared to the results of Adiprasito et 
al.~\cite{abm-nodim}, our radius bounds 
are slightly worse.  More precisely, they 
show that both in the colorful and the non-colorful
case, there is a ball of radius 
$O\left(\sqrt{k/n}\,\diam(P)\right)$ that intersects 
the convex hulls of the sets of the partition.
They also show this bound is close to optimal.
In contrast, our result is off by a factor of 
$O(\sqrt{k})$,
but derandomizing the proof of Adiprasito et al.~\cite{abm-nodim} 
gives only a brute-force $2^{O(n)}$-time algorithm. 
In contrast, our approach gives almost linear time algorithms 
for both cases, with a linear dependence on the dimension.

\paragraph{Techniques.}
Adiprasito et 
al.~first prove the colorful no-dimensional Tverberg 
theorem using an averaging argument over an exponential 
number of possible partitions. Then, they specialize their 
result for the non-colorful case, obtaining a bound that 
is asymptotically optimal. Unfortunately, it is not clear 
how to derandomize the averaging argument efficiently.
The method of conditional expectations applied to their
averaging argument leads to a running time of $2^{O(n)}$.
To get around this, we follow an alternate approach 
towards both versions of the  Tverberg theorem.
Instead of a direct averaging argument, we use a reduction 
to the Colorful \Caratheodory theorem that is inspired 
by Sarkaria's proof, with some additional twists.
We will see that this reduction also works in the no-dimensional 
setting, i.e., by a reduction to the no-dimensional 
Colorful \Caratheodory theorem of Adiprasito et al., we 
obtain a no-dimensional Tverberg theorem, with slightly weaker 
radius bounds, as stated above.
This approach has the advantage that
their Colorful \Caratheodory theorem 
is based on an averaging argument that 
permits an efficient derandomization using the 
method of conditional expectations~\cite{as-probabilistic}. 
In fact, we will see 
that the special structure of the no-dimensional Colorful \Caratheodory 
instance that we create allows for a very fast evaluation of the conditional 
expectations, as we fix the next part of the solution.
This results in an algorithm whose running time 
is $O(n d \ceil{\log k})$ instead of $O(n d k)$, as given by a 
naive application of the method. With a few interesting modifications,
this idea also works in the colorful setting.
This seems to be the first instance of using 
Sarkaria's method with special lifting vectors, and 
we hope that this will prove useful for further studies on
Tverberg's theorem and related problems.

\paragraph{Updates from the conference version.}
An extended abstract~\cite{cm-tverberg-socg} of this work appeared at the 
36th International Symposium on Computational Geometry.
The conference abstract omitted the details of the results
of Theorem~\ref{theorem:colorful-tverberg} and 
Theorem~\ref{theorem:nodim-gghs}.
In this version, we present all the missing details.

\paragraph{Outline of the paper.}
We describe our extension of 
Sarkaria's technique in Section~\ref{section:setup} 
and an averaging argument that is essential for our results.
In Section~\ref{section:theorem1}, we present the
proof of the no-dimensional Tverberg theorem (\Cref{theorem:tverberg}). 
The algorithm for computing the partition is also detailed therein.
Section~\ref{section:theorem2} contains the results for the
colorful setting of Tverberg (\Cref{theorem:colorful-tverberg}) and Section~\ref{section:theorem3}
presents results for the generalized Ham-Sandwich theorem (\Cref{theorem:nodim-gghs}).
We conclude in Section~\ref{section:conclusion} with some observations 
and open questions.

\section{Tensor product and Averaging argument}
\label{section:setup}

Let $P \subset \R^d$ be the 
given set of $n$ points.
We assume for simplicity that
the centroid of $P$, that we denote by $c(P)$,
coincides with the origin $\orig$, that is, 
$\sum_{x \in P} x = \orig$.
For ease of presentation, we denote the origin by
$\orig$ in all dimensions, as long as 
there is no danger of ambiguity.
Also, we write 
$\dotp{\cdot}{\cdot}$ for the usual scalar product
between two vectors in the appropriate dimension, and 
$[n]$ for the set $\{1, \dots, n\}$.

\subsection{Tensor product}
\label{section:tensor}

Let $x = (x_1, \dots, x_d) \in \R^d$ and 
$y = (y_1, \dots, y_m) \in \R^m$ be any two vectors.
The \emph{tensor product} $\otimes$ is 
the operation that takes $x$ and $y$ to 
the $dm$-dimensional vector $x\otimes y$
whose $ij$-th component is $x_i y_j$, that is, 
\[
x\otimes y= (xy_1, \dots, xy_m) =
\left(x_1y_1, \dots,x_dy_1, x_1y_2,\dots, x_d y_{m-1}, \dots,x_dy_m 
\right) \in \R^{dm}.
\]
Easy calculations show that for any 
$x, x'\in \R^d, y, y'\in \R^m$, the 
operator $\otimes$ satisfies:
\begin{enumerate}[(1)]
\item $x\otimes y + x'\otimes y = (x+x')\otimes y$;
\item $x\otimes y + x\otimes y' = x\otimes (y+y')$; and
\item\label{item:tensorscalar} $\dotp{x\otimes y}{x'\otimes y'} = \dotp{x}{x'}\dotp{y}{y'}$.
\end{enumerate}
By (\ref{item:tensorscalar}), 
the $L_2$-norm $\|x\otimes y \|$ of the tensor product 
$x\otimes y$ is exactly $\|x\|\|y\|$.
For any set of vectors $X = \{x_1, x_2, \dots\} $ in $\R^d$ and any
$m$-dimensional vector $q \in \R^m$, we denote by 
$X \otimes q$ the set of tensor products
$\{x_1 \otimes q, x_2 \otimes q, \dots \} \subset \R^{dm}$.
Throughout this paper, all distances will be measured  in the $L_2$-norm.

\paragraph{A set of lifting vectors.}
We generalize the tensor construction that was
used by Sarkaria to prove the Tverberg theorem~\cite{sarkaria}.
For this, we provide a way to construct a set of $k$ vectors 
$\{q_1, \dots, q_k\}$ that we use to create tensor products.
The motivation behind the precise choice of these vectors will be
clear in the next section, when we apply the construction to prove
the no-dimensional Tverberg result.
Let $\qgraph$ be an (undirected) simple, connected 
graph of $k$ nodes. 
Let
\begin{itemize}
\item $\numedges$ denote the number of edges in $\qgraph$, 

\item $\maxdegree$ denote the maximum degree of any node in $\qgraph$, and 

\item $\gdiam$ denote the diameter of $\qgraph$, i.e., the maximum length
of a shortest path between a pair of vertices in $\qgraph$.
\end{itemize}

We orient the edges of $\qgraph$ in an arbitrary manner
to obtain an oriented graph.
We use this directed version of $\qgraph$ to 
define a set of $k$ vectors $\{q_1, \dots,
q_k\}$ in $\numedges$ dimensions.
This is done as follows: each vector $q_i$ corresponds 
to a unique node $v_i$ of $\qgraph$ and its co-ordinates correspond
to the row in the oriented incidence matrix assigned to $v_i$.
More precisely, each coordinate position of the vectors 
corresponds to a unique edge of $\qgraph$. 
If $v_iv_j$ is a directed edge of $\qgraph$, then $q_i$ contains 
a $1$ and $q_j$ contains a $-1$ in the 
corresponding coordinate position.
The remaining co-ordinates are zero.
That means, the vectors $\{q_1, \dots, q_k\}$ are in $\R^{\numedges}$.
Also, $\sum_{i = 1}^{k} q_i = \orig$.
It can be verified that this is the unique linear dependence
(up to scaling) between the vectors for any choice of edge orientations
of $\qgraph$.
This means that the rank of the matrix with the $q_i$'s as the rows is $k-1$.
It can be verified that:
\begin{lemma}
\label{claim:basic-q}
For each vertex $v_i$, the squared norm $\|q_i\|^2$ is the degree of $v_i$.
For $i\neq j$, the dot product $\dotp{q_i}{q_j}$ 
is $-1$ if $v_iv_j$ is an edge in $\qgraph$, 
and $0$ otherwise.
\end{lemma}

An immediate application of Lemma~\ref{claim:basic-q} and 
property (\ref{item:tensorscalar}) of the tensor product is that for 
any set  of $k$ vectors $\{u_1,\dots,u_k\}$, each of the same dimension, 
the following relation holds:
\begin{align}
\left\|\sum_{i=1}^{k} u_i\otimes q_i \right\|^2
&=\sum_{i=1}^{k} \sum_{j=1}^{k} \dotp{u_i\otimes q_i}{u_j\otimes q_j}
\nonumber\\
&= \sum_{i=1}^{k} \sum_{j=1}^{k} \dotp{u_i}{u_j}\dotp{q_i}{q_j} 
\nonumber\\
&= \sum_{i=1}^{k} \dotp{u_i}{u_i}\dotp{q_i}{q_i}+
2\sum_{1\le i<j\le k}^{k} \dotp{u_i}{u_j}\dotp{q_i}{q_j}
\nonumber\\
&=\sum_{i=1}^{k} \|u_i\|^2 \|q_i\|^2-
2\sum_{v_iv_j\in E} \dotp{u_i}{u_j}
\nonumber\\
&=\sum_{v_iv_j\in E} \|u_i-u_j\|^2,
\label{equation:tree-product}
\end{align}
where $E$ is the set of edges of $\qgraph$.\footnote{We 
note that this identity is very similar to the Laplacian 
quadratic form that is used in spectral graph theory; see, 
e.g., the lecture notes by Spielman~\cite{spielman-sgt}
for more information.}

One of the simplest examples of such a set can be formed
by selecting $\qgraph$ to be the star graph.
Each of the $k-1$ leaves correspond to a standard basis vector of $\R^{k-1}$
and the root corresponds to $(-1,\dots,-1)\in \R^{k-1}$.
This is also the set used in \Barany and Onn's interpretation~\cite{bo-tverberg}
of Sarkaria's proof.

A more sophisticated example can be formed by 
taking $\qgraph$ as a balanced binary tree with $k$ nodes, 
and orienting the edges away from the root. 
Let $q_1$ correspond to the root.
A simple instance of the vectors is shown below:
\begin{center}
\begin{tikzpicture} 
[every node/.style = {minimum width = 2em, draw, circle},
level/.style = {sibling distance = 30mm/#1}]
\centering
\node{$q_1$}
child {
node {$q_2$} 
child {node {$q_4$}}
child {node {$q_5$}}
}
child {
node {$q_3$} 
child {node {$q_6$} }
child {node {$\dots$} }
};
\end{tikzpicture}
\end{center}
The vectors in the figure above can be represented as the matrix
\[
\begin{pmatrix}
q_1\\
q_2\\
q_3\\
q_4\\
q_5\\
q_6\\
\dots \\
\end{pmatrix}
=
\begin{pmatrix}
 1 & 1 & 0 & 0 & 0 & 0 & 0 & 0 &\dots \\
-1 & 0 & 1 & 1 & 0 & 0 & 0 & 0 &\dots  \\
 0 &-1 & 0 & 0 & 1 & 1 & 0 & 0 &\dots \\
 0 & 0 &-1 & 0 & 0 & 0 & 1 & 1 &\dots \\
 0 & 0 & 0 &-1 & 0 & 0 & 0 & 0 & \dots \\
 0 & 0 & 0 & 0 & -1 & 0 & 0 & 0 & \dots \\
 &&&\dots &&&
\end{pmatrix}
\]
where the $i$-th row of the matrix corresponds to vector $q_i$.
As $\numedges=k-1$, each vector is in $\R^{k-1}$.
The norm $\|q_i\|$ is either 
$\sqrt{2}$, $\sqrt{3}$, or $1$, depending 
on whether $v_i$ is the root, an internal node with two children, 
or a leaf, respectively.
The height of $\qgraph$ is $\ceil {\log k}$ 
and the maximum degree is $\maxdegree=3$.

\subsection{Averaging argument}
\label{subsection:averaging}

\paragraph{Lifting the point set.}
Let $ P = \{p_1, \dots, p_n\}\subset \R^d$.
We first pick a graph $\qgraph$ with $k$ vertices, as 
in the previous paragraph, and we derive a set of $k$ 
lifting vectors $\{q_1, \dots, q_k\}$ from $\qgraph$.
Then, we lift each point of $P$ to a set of vectors 
in $d\numedges$ dimensions, by taking
tensor products with the vectors $\{q_1, \dots, q_k\}$.
More precisely, for $a \in [n]$ and 
$j \in [k]$, let 
$
p_{a,j} = p_a \otimes q_j \in \R^{d\numedges}.
$
For $a \in [n]$, we let 
$P_a = \{p_{a,1}, \dots,p_{a,k}\}$ be the 
lifted points obtained from $p_a$.
We have, 
$
\|p_{a,j}\| = \|q_j\| \|p_a\| \leq \sqrt{\maxdegree}\|p_a\|.
$
By the bi-linear properties of the tensor
product, we have
\[
c(P_a) = \frac{1}{k} \sum_{j = 1}^{k} \left( p_a  \otimes q_j\right)
= \frac{1}{k} \left( p_a \otimes \left(\sum_{j=1}^{k}  q_j \right)  \right)
= \frac{1}{k} \left( p_a \otimes \orig \right)
= \orig,
\]
so the centroid $c(P_a)$ coincides with the origin, for $a \in [n]$.

The next lemma contains the technical core of our argument. 
The result is applied in Section~\ref{section:theorem1}
to derive a useful partition 
of $P$ into $k$ subsets of prescribed sizes from the lifted point sets.
\begin{lemma}
\label{lemma:balanced-traversal}
Let $P = \{p_1, \dots, p_n \}$ be a set of $n$ points 
in $\R^d$ satisfying $\sum_{i=1}^{n} p_i = \orig$. 
Let $P_1, \dots, P_n$ denote the point sets 
obtained by lifting each $p_i \in P$ using the vectors 
$\{ q_1, \dots, q_k \}$ defined using a graph $\qgraph$.

\begin{enumerate}[(i)]
\item
\label{avg1} 
For any choice of 
positive integers $r_1,\dots,r_k$ that satisfy $\sum_{i=1}^{k} r_i =n$,
there is a partition $T_1, \dots, T_k$ of $P$ with 
$|T_1|=r_1, |T_2|=r_2, \dots , |T_k|=r_k$ such that the centroid of 
the set of lifted points 
$T:= T_1 \otimes q_1 \cup \dots \cup T_k \otimes q_k$
(this set is also a traversal of $P_1, \dots, P_n$)
has distance less than 
\[
\delta=\sqrt{\frac{ \maxdegree}{2(n-1)}}\diam(P)
\]
from the origin $\orig$.

\item 
\label{avg2}
The bound is better for the case $n=rk$ and $r_1=\dots=r_k= n/k$.
There exists a partition $T_1, \dots, T_k$ of $P$ with 
$|T_1| = |T_2| = \dots = |T_k| = r$ such that the centroid of 
$T:= T_1 \otimes q_1 \cup \dots \cup T_k \otimes q_k$
has distance less than 
\[
\gamma = \sqrt{\frac{\numedges}{k(n-1)}}\diam(P)
\]
from the origin $\orig$.\end{enumerate}

\end{lemma}

\begin{proof}
We use an averaging argument to prove the claims,
like Adiprasito et al.~\cite{abm-nodim}.
More precisely, we bound the average norm $\delta$ of the centroid of the 
lifted points $T_1 \otimes q_1 \cup \dots 
\cup T_k \otimes q_k $ over 
all partitions of $P$ of the form $T_1, \dots, T_k$,
for which the sets in the partition have sizes $r_1, \dots, r_k$ respectively, 
with $\sum_{i=1}^{k} r_i = n$. 

\paragraph{Proof of ~\Cref{lemma:balanced-traversal}(\ref{avg1}).}
Each such partition can be interpreted as a traversal 
of the lifted point sets $P_1, \dots, P_n$ that contains $r_i$ points
lifted with $q_i$, for $i \in [k]$.
Thus, consider any traversal of this type $X = \{x_1, \dots, x_n\}$
of $P_1,\dots,P_n $, where $x_a \in P_a$, 
for $a \in [n]$.
The centroid of $X$ is 
$c(X) = (1/n)\sum_{a=1}^{n} x_a$.
We bound the expectation $n^2\av\left(\|c(X)\|^2\right) = 
\av\left(\left\|\sum_{a = 1}^n x_a\right\|^2\right)$, 
over all possible traversals $X$.
By the linearity of expectation, 
$\av\left(\left\|\sum_{a = 1}^n x_a\right\|^2\right)$
can be written as
\[
\av\left(\left\|\sum_{a = 1}^n x_a\right\|^2\right)
=\av\left( \sum_{a = 1}^n \|x_a\|^2  
+\sum_{\substack{a,b \in [n]\\ a < b}} 2\dotp{x_a}{x_b}\right)
=\av\left( \sum_{a = 1}^n \|x_a\|^2\right)
+2\av\left(\sum_{\substack{a,b \in [n]\\ a < b}} \dotp{x_a}{x_b}\right).
\]
We next find the coefficient of each term of the form
$\|x_a\|^2$ and $\dotp{x_a}{x_b}$ 
in the expectation.
Using the multinomial coefficient, the total number 
of traversals $X$ is
\[
\begin{pmatrix} n \\ r_1, r_2, \dots, r_k \end{pmatrix} = 
\frac{n!}{r_1!r_2! \cdot \dots \cdot r_k!}.
\]
Furthermore, for any lifted point $x_a=p_{a,j}$, the number of 
traversals $X$ with $p_{a,j} \in X$ is
\[
\begin{pmatrix}n - 1 \\ r_1, \dots, r_j - 1, \dots, r_k\end{pmatrix}
= \frac{(n - 1)!}{r_1! \cdot \dots \cdot (r_{j} - 1)! \cdot \dots \cdot r_k!}.
\]
So the coefficient of $\|p_{a,j}\|^2$ is 
\[
\frac{\frac{(n - 1)!}{r_1! \cdot \dots \cdot (r_{j} - 1)! \cdot 
  \dots \cdot r_k}}
{\frac{n!}{r_1! \cdot \dots \cdot r_k!}} = \frac{r_j}{n}.
\]
Similarly, for any pair of points $(x_a,x_b)=(p_{a,i},p_{b,j})$, 
there are two cases
in which they appear in the same traversal:
first, if $i = j$, the number of traversals is 
\[
\frac{(n - 2)!}{r_1! \cdot \dots \cdot (r_{i} - 2)! \cdot \dots \cdot r_k!}.
\]
The coefficient of $\dotp{p_{a,i}}{p_{b,j}}$ 
in the expectation is hence 
\[\frac{r_i(r_i - 1)}{n(n - 1)}.
\]
Second, if $i\neq j$, the number of traversals is calculated to be
\[\frac{(n - 2)!}{r_1! \cdot \dots \cdot (r_{i} - 1)! \cdot \dots \cdot 
  (r_j - 1)! \cdot \dots \cdot r_k!}.\]
The coefficient of $\dotp{p_{a,i}}{p_{b,j}}$ 
is 
\[\frac{r_i r_j}{n(n - 1)}.
\]
Substituting the coefficients, we bound
the expectation as
\begin{align*}
&\av\left( \sum_{a = 1}^n \|x_a\|^2\right)
+2\av\left(\sum_{\substack{a,b \in [n]\\ a < b}} \dotp{x_a}{x_b}\right)
\\
&= \sum_{a=1}^{n} \sum_{j=1}^{k}\| p_{a,j}\|^2 \frac{r_j}{n}
+2\sum_{\substack{a,b \in [n]\\ a < b}}
\left(\sum_{j = 1}^{k}\langle p_{a,j}, p_{b,j} \rangle \frac{r_j(r_{j} - 1)}{n(n - 1)}
+\sum_{\substack{i,j \in [k]\\ i \neq j}}  \langle p_{a,i}, p_{b,j} \rangle 
\frac{r_{i}r_{j}}{n(n-1)}
\right)
\\
&=\sum_{j=1}^{k}\frac{r_j}{n}\sum_{a=1}^{n}\| p_{a,j}\|^2
+ \frac{2}{n(n -1)}\sum_{\substack{a,b \in [n]\\ a < b}}
\left(\sum_{i, j \in [k]}  \langle p_{a,i}, p_{b,j} \rangle 
  r_{i}r_{j}
-\sum_{j=1}^{k} \langle p_{a,j}, p_{b,j} \rangle r_{j}
\right)
\\
&=\sum_{j=1}^{k}r_j \left(\frac{1}{n}\sum_{a=1}^{n}\|p_{a,j}\|^2\right)
+\sum_{\substack{a,b \in [n]\\ a < b}}
\sum_{i, j \in [k]} \frac{2\dotp{p_{a,i}}{p_{b,j}} r_{i}r_{j}}{n(n-1)}
-\sum_{\substack{a,b \in [n]\\ a < b}}
\sum_{j=1}^{k} \frac{2\dotp{p_{a,j}}{p_{b,j}}r_{j}}{n(n-1)}.
\end{align*}
We bound the value of each of the three terms individually to get 
an upper bound on the value of the expression.
The first term can be bounded as
\begin{align*}
\sum_{j=1}^{k}r_j \left(\frac{1}{n}\sum_{a=1}^{n}\|p_{a,j}\|^2\right)
=&\frac{1}{n}\sum_{j=1}^{k}r_j
\left(\sum_{a=1}^{n}\|p_{a}\|^2\|q_j\|^2\right)
\\
=&\frac{1}{n}\left(\sum_{j=1}^{k}r_j\|q_j\|^2\right)
\sum_{a=1}^{n}\|p_{a}\|^2 
\\
\le &\frac{1}{n}\left(\maxdegree\sum_{j=1}^{k}r_j\right)
\sum_{a=1}^{n}\|p_{a}\|^2
\\
=& \frac{1}{n}\left(\maxdegree n\right)
\sum_{a=1}^{n}\|p_{a}\|^2
\\
<&\maxdegree\left(\frac{n\diam(P)^2}{2}\right), 
\end{align*}
where we have made use of Lemma~\ref{claim:basic-q} and the fact that 
$\sum_{a=1}^{n}\|p_{a}\|^2<\frac{n\diam(P)^2}{2}$ 
(see~\cite[Lemma\,~4.1]{abm-nodim}). 
The second term can be re-written as
\begin{align*}
\sum_{\substack{a,b \in [n]\\ a < b}}
\sum_{i, j \in [k]}  \frac{2 \dotp{p_{a,i}}{p_{b,j}} r_i r_j}{n(n-1)} 
=&
\sum_{i, j \in [k]}  
\frac{2 r_i r_j}{n(n-1)}\left(  \sum_{\substack{a,b \in [n]\\ a < b}}
\dotp{p_{a,i}}{p_{b,j}} \right)
\\
=&
\sum_{i, j \in [k]} \frac{2 r_i r_j}{n(n-1)}\left( \sum_{\substack{a,b \in [n]\\ a < b}} 
\left\langle  p_a \otimes q_i, p_b \otimes q_j \right\rangle 
\right)
\\
=&
\sum_{i, j \in [k]}  \frac{2 r_i r_j}{n(n-1)}\left( \sum_{\substack{a,b \in [n]\\ a < b}} 
\left\langle p_a, p_b\right\rangle
\dotp{q_i}{q_j}
\right)
\\
=&
\left( \sum_{i, j \in [k]} \frac{2 \dotp{q_i}{q_j} r_i r_j}{n(n-1)}  \right)
\sum_{\substack{a,b \in [n]\\ a < b}} \dotp{p_a}{p_b} 
\\
=&
\frac{2}{n(n-1)}\left( \sum_{i, j \in [k]} \dotp{q_i}{q_j} r_i r_j  \right)
\sum_{\substack{a,b \in [n]\\ a < b}} \dotp{p_a}{p_b}.
\end{align*}
The expression $\sum_{i, j \in [k]} \dotp{q_i}{q_j} r_i r_j$ can be further simplified as
\begin{align*}
\sum_{i, j \in [k]} \dotp{q_i}{q_j} r_i r_j 
&= \sum_{1\le i=j \le k} \dotp{q_i}{q_j} r_i r_j + 2\left( \sum_{1\le i<j \le k} \dotp{q_i}{q_j} r_i r_j \right) \\
&= \sum_{1\le i \le k} \|q_i\| r_i^2 + 2\left(\sum_{v_iv_j\in E} (-1)\cdot r_ir_j + \sum_{v_iv_j\not\in E} 0\cdot r_ir_j\right)\\
&= \sum_{1\le i \le k} \mathrm{degree}(v_i) r_i^2 + \sum_{v_iv_j\in E} -2r_ir_j\\
&= \sum_{v_iv_j\in E} r_i^2 + r_j^2 -2r_ir_j\\
&= \sum_{v_iv_j\in E} (r_i-r_j)^2.
\end{align*}
where we have again made use of Lemma~\ref{claim:basic-q}.
Substituting, the second term becomes 
\[
\frac{2}{n(n-1)}\left( \sum_{(v_i,v_j)\in E} (r_i-r_j)^2  \right)
\sum_{\substack{a,b \in [n]\\ a < b}} \dotp{p_a}{p_b}<0,
\]
since we can use $c(P)=\orig$ to bound 
$
\sum_{a, b, \in [n], a < b} \dotp{p_a}{p_b} = -\frac{1}{2}\sum_{a=1}^{n} \|p_a\|^2 <0.
$
The second term is non-positive and therefore 
can be removed since the total expectation
is always non-negative.
The third term is
\begin{align*}
\sum_{\substack{a,b \in [n]\\ a < b}}
\sum_{j=1}^{k}\frac{-2\dotp{p_{a,j}}{p_{b,j}}r_{j}}{n(n-1)}
&=\sum_{\substack{a,b \in [n]\\ a < b}} \sum_{j=1}^{k}
\frac{-2\left\langle p_a\otimes q_j, p_b\otimes q_j \right \rangle r_{j}}{n(n-1)}
\\
&=\sum_{\substack{a,b \in [n]\\ a < b}}\sum_{j=1}^{k}
\frac{-2\dotp{p_a}{p_b}\|q_j\|^2 r_j }{n(n-1)}
\\
&= \left(\sum_{j=1}^{k}\|q_j\|^2 r_j\right)
\left(\sum_{\substack{a,b \in [n]\\ a < b}} 
  \frac{-2\dotp{p_a}{p_b}}{n(n-1)}
\right)
\\
&< \left(\sum_{j=1}^{k}\|q_j\|^2 r_j\right)
\left( \frac{n\diam(P)^2}{2n(n-1)}
\right)
\\
&= \left(\sum_{j=1}^{k}\|q_j\|^2 r_j\right)
\frac{\diam(P)^2}{2(n-1)}
\\
&< \frac{n\maxdegree\diam(P)^2}{2(n-1)}.
\end{align*}
Collecting the three terms, the expression is upper bounded by
\[
\frac{\diam(P)^2 \maxdegree n }{2}  
+ \frac{\diam(P)^2\maxdegree n}{2(n-1)}
=\frac{\diam(P)^2 \maxdegree n }{2} \left(1+\frac{1}{n-1}\right)
=\frac{\diam(P)^2 \maxdegree n^2 }{2(n-1)}  ,
\]
which bounds the expectation by
\begin{align*}
\frac{1}{n^2}\left(\frac{\diam(P)^2 \maxdegree n^2 }{2(n-1)}
\right)
=\frac{\diam(P)^2 \maxdegree}{2(n-1)}.
\end{align*}
This shows that there is a traversal 
such that its centroid has norm less than 
\[
\diam(P)\sqrt{\frac{ \maxdegree}{2(n-1)}}.
\]

\paragraph{Proof of ~\Cref{lemma:balanced-traversal}(\ref{avg2}) (balanced case).}
For the case that $n$ is a multiple of $k$,
and $r_1=\dots=r_k=\frac{n}{k}=r$, the upper bound can be improved:
the first term in the expectation is 
\begin{align*}
\sum_{j=1}^{k}r_j \left(\frac{1}{n}\sum_{a=1}^{n}\|p_{a,j}\|^2\right)
&=\frac{r}{n}\sum_{j=1}^{k} \sum_{a=1}^{n}\|p_{a,j}\|^2
\\
&=\frac{r}{n}\sum_{j=1}^{k} \sum_{a=1}^{n} \|p_{a}\|^2 \|q_j\|^2
\\
&=\frac{r}{n}\left(\sum_{j=1}^{k}\|q_j\|^2\right) 
\sum_{a=1}^{n} \|p_{a}\|^2
\\
&=\frac{r}{n}2\numedges\sum_{a=1}^{n} \|p_{a}\|^2
\\
&<\frac{r}{n}2\numedges\left(\frac{n\diam(P)^2}{2} \right)
\\
&\le r\numedges\diam(P)^2,
\end{align*}
The second term is zero, and
the third term is less than
\begin{align*}
\left(\sum_{j=1}^{k}\|q_j\|^2 r_j\right) \frac{\diam(P)^2}{2(n-1)}
&=r\left(\sum_{j=1}^{k}\|q_j\|^2 \right)\frac{\diam(P)^2}{2(n-1)}
\\
&=2r\numedges\frac{\diam(P)^2}{2(n-1)}
\\
&=\frac{r\numedges\diam(P)^2}{(n-1)}.
\end{align*}
The expectation is upper bounded as 
\begin{align*}
n^2\av\left(\|c(X)\|^2\right) &<  r\numedges\diam(P)^2 
+ \frac{r\numedges\diam(P)^2}{(n-1)}
\\
\implies \av\left(\|c(X)\|^2\right) &< 
\frac{r\numedges\diam(P)^2 }{n^2}\left(1 + \frac{1}{n-1} \right)
\\
&=\frac{r\numedges\diam(P)^2}{n(n-1)}
=\frac{\numedges\diam(P)^2}{k(n-1)},
\end{align*}
which shows that there is at least one balanced traversal $X$ whose centroid
has norm less than 
\[\sqrt{\frac{\numedges}{k(n-1)}}\diam(P),\]
as claimed.
\end{proof}

\section{Efficient no-dimensional Tverberg Theorem}
\label{section:theorem1}

In this section we prove the results of \Cref{theorem:tverberg}:
\ndtverberg*

\subsection{Proof of \Cref{theorem:tverberg}(\ref{tverberg1})}
\label{subsection:tverberg1}

We lift the points of $P$ to $P_1,\dots,P_n$ using a graph $\qgraph$ and 
the associated vectors $q_1,\dots,q_k$ as in Section~\ref{subsection:averaging}.
The centroid $c(P_a)$ coincides with the origin, for $a \in [n]$.
Applying Lemma~\ref{lemma:balanced-traversal}, there is a traversal 
$T:= T_1 \otimes q_1 \cup \dots \cup T_k \otimes q_k$ of the lifted
points, with $|T_1|=r_1, |T_2|=r_2, \dots , |T_k|=r_k$,
such that its centroid has norm at most $\delta$.

We show that there is a ball of bounded radius 
that intersects the convex hull of each $T_i$.
Let $\alpha_1=r_1/n,\dots,\alpha_k=r_k/n$ be 
positive real numbers.
The centroid of $T$, $c(T)$, can be written as 
\[
c(T) = \frac{1}{n}\sum_{i=1}^{k} \sum_{p\in T_i} p \otimes q_i
=\sum_{i=1}^{k} \frac{1}{n}\left(\sum_{p \in T_i} p \right)\otimes q_i
=\sum_{i=1}^{k} \frac{r_i}{n}\left(\frac{1}{r_i}\sum_{p \in T_i} p \right)\otimes q_i
=\sum_{i=1}^{k} \alpha_i c_i\otimes q_i,
\]
where $c_i = c(T_i)$ denotes the centroid of $T_i$, for $i \in [k]$.
Using Equation~\eqref{equation:tree-product},
\begin{align}
\|c(T)\|^2 &= 
\left\| \sum_{i=1}^{k} \alpha_i c_i\otimes q_i \right\|^2
=\sum_{v_i v_j\in E} \|\alpha_i c_i - \alpha_j c_j\|^2.
\label{equation:centroid-norm}
\end{align}
Let $x_1=\alpha_1 c_1, x_2=\alpha_2 c_2, \dots, x_k=\alpha_k c_k$.
Then,
\[
\sum_{i=1}^{k} x_i = \sum_{i=1}^{k} \alpha_i c_i =  
\sum_{i=1}^{k} \frac{r_i}{n}\left( \frac{1}{r_i}\sum_{p\in T_i} p\right)
=\frac{1}{n}\sum_{j=1}^{n} p_j = \orig,
\]
so the centroid of $\{ x_1,\dots,x_k \}$ coincides with the origin.
Using  $\| c(T) \| < \delta$ and Equation~\eqref{equation:centroid-norm}, 
\begin{align*}
\sum_{v_i v_j\in E} \|x_i - x_j\|^2
=\sum_{v_i v_j\in E} \|\alpha_i c_i - \alpha_j c_j\|^2
<\delta^2.
\end{align*}

We bound the distance from $x_1$ to every other $x_i$.
For each $i\in [k]$, we associate to $x_i$ the node $v_i$  in $\qgraph$.
Let the shortest path from $v_1$ to $v_j$ in $\qgraph$ be denoted by 
$(v_1, v_{i_1}, v_{i_2}, \dots, v_{i_z}, v_{j})$. 
This path has length at most $\gdiam$.
Using the triangle inequality and the Cauchy-Schwarz inequality,

\begin{align}
\|x_1 - x_j\| &\leq
\|x_1 - x_{i_1}\| + \|x_{i_1} - x_{i_2}\| + \dots + \|x_{i_z} - x_j\|\nonumber\\
&\leq \sqrt{\gdiam} \sqrt{\|x_1 - x_{i_1}\|^2 + 
\|x_{i_1} - x_{i_2}\|^2 + \dots + \|x_{i_z} - x_j\|^2}\nonumber\\
& \leq \sqrt{\gdiam}\sqrt{ \sum_{v_iv_j \in E} \|x_i - x_j\|^2}
< \sqrt{\gdiam}\delta.
\label{equation:gen-radius-bound}
\end{align}
Therefore, the ball of radius $\beta:=\sqrt{\gdiam}\delta$ centered at $x_1$
covers the set $\{x_1, \dots, x_k\}$.
That means, the ball covers the convex hull of $\{x_1, \dots, x_k\}$
and in particular contains the origin.
Using the triangle inequality, the ball of radius $2\beta$ 
centered at the origin contains $\{x_1, \dots, x_k\}$.
Then, the norm of each $x_i$ is at most $2\beta$,
which implies that the norm of each $c_i$ is at most $2\beta/\alpha_i$.
Therefore, the ball of radius 
\[
\frac{2\beta}{\mathrm{min}_i \alpha_i} 
= \frac{2n\sqrt{\gdiam}\delta}{\mathrm{min}_i r_i}
\]
centered at $\orig$ contains the set $\{ c_1,\dots,c_k \}$.
Substituting the value of $\delta$ from 
Lemma~\ref{lemma:balanced-traversal},  
the ball of radius
\[
\frac{2n\sqrt{\gdiam}}{\mathrm{min}_i r_i} 
\sqrt{\frac{\maxdegree}{2(n-1)}}\diam(P) 
=\frac{n\diam(P)}{\mathrm{min}_i r_i} \sqrt{\frac{2\gdiam\maxdegree}{n-1}}
\]
centered at $\orig$ covers the set $\{ c_1,\dots,c_k \}$.

\paragraph{Optimizing the choice of $\qgraph$.}
The radius of the ball has a term $\sqrt{\gdiam \maxdegree}$ 
that depends on the choice of $\qgraph$.
For a path graph this term has value $\sqrt{(k-1)2}$. For a star graph,
that is, a tree with one root and $k-1$ children,
this is $\sqrt{k-1}$.
If $\qgraph$ is a balanced $s$-ary tree, then the Cauchy-Schwarz inequality
in Equation~\eqref{equation:gen-radius-bound}
can be modified to replace $\gdiam$ by the height of the tree.
Then, the term is $\sqrt{\ceil{\log_s k}(s+1)}$,
which is minimized for $s=4$.
For this choice of $\qgraph$, the radius is bounded by 
\[
\frac{n\diam(P)}{\mathrm{min}_i r_i} \sqrt{\frac{10\ceil{\log_4 k} }{n-1}},
\] 
as claimed.
\qed

\subsection{Proof of \Cref{theorem:tverberg}(\ref{tverberg2}) (balanced partition)}
\label{subsection:tverberg2}
For the case $n=rk$ and $r_1=\dots=r_k=r$, we give a better
bound for the radius of the ball containing the centroids 
$c_1,\dots,c_k$.
In this case, 
we have $\alpha_1=\alpha_2=\dots=\alpha_k=r/n=1/k$.
Then, Equation~\eqref{equation:centroid-norm} is
\begin{align*}
\|c(T)\|^2=\sum_{v_i v_j\in E} \|\alpha_i c_i - \alpha_j c_j\|^2
=\frac{1}{k^2} \sum_{v_i v_j\in E} \| c_i -  c_j\|^2.
\end{align*}
Since $\| c(T) \| < \gamma$, we get
\begin{equation}
\sum_{v_i v_j\in E} \| c_i -  c_j\|^2 < k^2\gamma^2.
\label{equation:balanced-centroid-norm}
\end{equation}
Similar to the general case, we bound the distance from $c_1$ 
to any other centroid $c_j$.
For each $i$, we associate to $c_i$ the node $v_i$ in $\qgraph$.
There is a path of length at most $\gdiam$ from $v_1$ to any other
node.
Using the Cauchy-Schwarz inequality and substituting the value of $\gamma$ 
from Lemma~\ref{lemma:balanced-traversal}, 
we get
\begin{align}
\|c_1-c_j\|&\le \sqrt{\gdiam}
\sqrt{\sum_{v_i v_j\in E} \| c_i -  c_j\|^2 }
<\sqrt{\gdiam}k\gamma 
\nonumber\\
& = \sqrt{\frac{\gdiam \numedges}{k(n-1)}}k\diam(P)
\\ 
&=\sqrt{\frac{k}{n-1}}\sqrt{\gdiam \numedges}\diam(P).
\label{equation:balanced-radius-bound}
\end{align}
Therefore, a ball of radius 
\[
\sqrt{\frac{k}{n-1}}\sqrt{\gdiam \numedges}\diam(P)\]
centered at $c_1$
contains the set $c_1,\dots,c_k$.
The factor $\sqrt{\gdiam \numedges}$
is minimized when $\qgraph$ is a star graph, which is a tree.
We can replace the term $\gdiam$ by the height of the tree.
Then, the ball containing  $c_1,\dots,c_k$ has radius
\[
\sqrt{\frac{k(k-1)}{n-1}}\diam(P),
\]
as claimed. 
\qed

\paragraph{As balanced as possible.}
When $k$ does not divide $n$, but we still want a balanced partition,
we take any subset of $n_0=k\floor{n/k}$ points of $P$
and get a balanced Tverberg partition on the subset.
Then, we add the removed points one by one to the sets of the partition,
adding at most one point to each set.
As shown above, there is a ball of radius less than 
\[
\sqrt{\frac{k(k-1)}{n_0-1}}\diam(P)
\]
that intersects the convex hull of each set in the partition.
Noting that 
\[
\frac{1}{\sqrt{n_0 - 1}} \leq 
\sqrt{\frac{k+2}{k}}\frac{1}{\sqrt{n-1}},
\]
a ball of radius less than 
\[
\sqrt{\frac{(k+2)(k-1)}{(n-1)}}\diam(P)
\]
intersects the convex hull of each set of the partition.

\subsection{Proof of \Cref{theorem:tverberg}(\ref{tverberg3})(computing the Tverberg partition)}
\label{subsection:tverberg3}

We now give a deterministic algorithm to compute 
no-dimensional Tverberg partition $T_1,\dots,T_k$.
The algorithm is based on the method of conditional expectations.
First, in Section~\ref{subsubsection:algo-general} 
we give an algorithm for the general case
when the sets in the partitions are constrained to have
given sizes $r_1,\dots,r_k$.
The choice of $\qgraph$ is crucial for the algorithm.

The balanced case of $r_1=\dots=r_k$ has a better radius bound
and uses a different graph $\qgraph$.
The algorithm for the general case also extends to the balanced case 
with a small modification, that we discuss 
in Section~\ref{subsubsection:algo-balanced}.
We get the same runtime in either case.

\subsubsection{Algorithm for the general case}
\label{subsubsection:algo-general}

As before, the input is a set of $n$ points $P\subset \R^d$ and $k$ positive integers
$r_1,\dots,r_k$ satisfying $\sum_{i=1}^{k} r_i=n$.
Using tensor product construction, 
each point of $P$ is lifted implicitly using the vectors $\{q_1,\dots,q_k\}$ 
to get the set $\{P_1,\dots,P_n\}$.
We then compute the required traversal of $\{P_1,\dots,P_n\}$ using the method
of conditional expectations~\cite{as-probabilistic}, 
the details of which can be found below.
Grouping the points of the traversal according to the lifting vectors
used gives us the required partition.
We remark that in our algorithm, we do not explicitly lift any vector
using the tensor product, thereby avoiding costs associated with working 
on vectors in $d\numedges$ dimensions.

We now describe a procedure to find a traversal that corresponds
to a desired partition of $P$.
We go over the points in $\{P_1,\dots,P_n \}$ iteratively in reverse order
and find the 
traversal $Y=(y_1\in P_1,\dots,y_n\in P_n)$ point by point.
More precisely, we determine $y_n$ in the first step, 
then $y_{n-1}$ in the second step, and so on.
In the first step, we go over all points of $P_n$ and select any point 
$y_n\in P_n$ that satisfies 
$\av\left(\|c(x_1,x_2,\dots,x_{n-1},y_n)\|^2\right)
\le \av\left(\|c(x_1,x_2,\dots,x_{n-1},x_n)\|^2\right)$.
For the general step, suppose we have already selected the points $\{ y_{s+1},y_{s+2},\dots,y_n \}$.
To determine $y_s$, we choose any point from $P_s$ that achieves
\begin{equation}
\av\left(\|c(x_1,\dots,x_{s-1},y_{s},y_{s+1},\dots,y_n)\|^2\right)\le 
\av\left(\|c(x_1,\dots,x_{s},y_{s+1},\dots,y_n)\|^2\right).
\label{equation:conditional-ex}
\end{equation}
The last step gives the required traversal.
We expand the expectation as 
\begin{align*}
&\av(\|c(x_1,x_2,\dots,x_{s-1},y_{s},\dots,y_n)\|^2)\\
&=\av\left(\left\| \frac{1}{n}
\left(\sum_{i=1}^{s-1}x_i + \sum_{i=s}^{n}y_i \right)\right\|^2\right)
=\frac{1}{n^2}\av\left( \left\|\left(\sum_{i=1}^{s-1}x_i + \sum_{i=s+1}^{n}y_i\right)+y_s\right\|^2 \right)
\\
&=
\frac{1}{n^2}\left(
\av\left( \left\| \sum_{i=1}^{s-1}x_i + \sum_{i=s+1}^{n}y_i\right\|^2 \right) 
+ \|y_s\|^2 + 2\left\langle y_s,
\av\left(\sum_{i=1}^{s-1}x_i + \sum_{i=s+1}^{n}y_i\right) 
\right\rangle \right)
\\
&=
\frac{1}{n^2}\left(
\av\left( \left\|\sum_{i=1}^{s-1}x_i + \sum_{i=s+1}^{n}y_i\right\|^2 \right)
+ \|y_s\|^2 + 2\left\langle y_s,
\av\left(\sum_{i=1}^{s-1}x_i\right) + \sum_{i=s+1}^{n}y_i 
\right\rangle \right).
\end{align*}
We pick a $y_s$ for which  
$\av(\left\|c(x_1,x_2,\dots,x_{s-1},y_{s},\dots,y_n)\right\|^2)$ 
is at most the average over all choices of $y_s\in P_s$.
As the term
$\av\left( \left\|\sum_{i=1}^{s-1}x_i + \sum_{i=s+1}^{n}y_i\right\|^2 \right)$ 
is constant 
over all choices of $y_s$, and the factor $\frac{1}{n^2}$ 
is constant, we can remove them from consideration.
We are left with
\begin{equation}
\|y_s\|^2 + 2\left\langle y_s,
\av\left(\sum_{i=1}^{s-1}x_i\right) + \sum_{i=s+1}^{n}y_i 
\right\rangle
=
\|y_s\|^2 + 2\left\langle y_s,
\av\left(\sum_{i=1}^{s-1}x_i\right)\right\rangle 
+ 2\dotp{y_s}{\sum_{i=s+1}^{n}y_i}. 
\label{equation:expectedsum}
\end{equation}
Let $y_s = p_s \otimes q_i$ without loss of generality.
The first term is 
\[
\|y_s\|^2= \|p_s \otimes q_i\|^2 = \|p_s\|^2\|q_i\|^2.
\]
Let $r'_1,\dots,r'_k$ be the number of elements of $T_1,\dots,T_k$
that are yet to be determined. 
In the beginning, $r'_i=r_i$ for each $i$.
Using the coefficients from Section~\ref{subsection:averaging},
$\av\left(\sum_{i=1}^{s-1}x_i\right)$ can be written as 
\begin{align*}
\av\left(\sum_{i=1}^{s-1}x_i\right)
&=\sum_{i=1}^{s-1}\sum_{j=1}^{k}p_{i,j}\frac{r'_j}{s-1}
\\
&=\sum_{j=1}^{k}\frac{r'_j}{s-1}\sum_{i=1}^{s-1}p_{i,j}
\\
&=\sum_{j=1}^{k}\frac{r'_j}{s-1}\sum_{i=1}^{s-1} p_i \otimes q_j
\\
&=\frac{1}{s-1}\sum_{j=1}^{k}r'_j \left(\sum_{i=1}^{s-1} p_i\right)
\otimes q_j
\\
&=\left(\frac{1}{s-1}\sum_{i=1}^{s-1} p_i\right)
\otimes \left(\sum_{j=1}^{k}r'_j  q_j\right)
\\
&=c_{s-1} \otimes \left(\sum_{j=1}^{k}r'_j  q_j\right),
\end{align*}
where $c_{s-1}=\frac{\sum_{i=1}^{s-1} p_i}{s-1}$ is the centroid of 
the first $(s-1)$ points.
Using this, the second term can be simplified as
\begin{align*}
2\left\langle y_s,\av\left(\sum_{i=1}^{s-1}x_i\right)\right\rangle
&=2 \left\langle p_s \otimes q_i, c_{s-1} \otimes 
\left(\sum_{j=1}^{k}r'_j q_j\right) \right\rangle
\\
&=2 \left\langle p_s, c_{s-1} \right\rangle
\left\langle q_i, \sum_{j=1}^{k}r'_j q_j \right\rangle
\\
&=2\dotp{p_s}{c_{s-1}} 
\left(r'_i\|q_i\|^2-\sum_{v_iv_j\in E}r'_j\right)
\\
&=\dotp{p_s}{c_{s-1}} R_i,
\end{align*}
where $R_i=2\left(r'_i\|q_i\|^2-\sum_{v_iv_j\in E} r'_j\right)$.
The third term is $2\left\langle y_s,\sum_{j=s+1}^{n}y_j \right\rangle$.
Let $y_j=p_j\otimes q_{m_j}$ for $s+1\le j \le n$.
The term can be simplified to
\begin{align*}
2\left\langle y_s,\sum_{j=s+1}^{n}y_j \right\rangle
&=2\sum_{j=s+1}^{n}\dotp{y_s}{y_j}
\\
&=2\sum_{j=s+1}^{n} \left\langle
p_s  \otimes q_i, p_j \otimes q_{m_j} \right\rangle
\\
&=2\sum_{j=s+1}^{n} \dotp{p_s}{p_j} \dotp{q_i}{q_{m_j}}
\\
&= 2 \left\langle p_s, \sum_{p\in T_i} p \|q_i\|^2 - 
\sum_{j: v_i v_j\in E} \sum_{p\in T_j} p
\right\rangle 
\\
&=\left\langle p_s, 2\left( \|q_i\|^2 \sum_{p\in T_i} p - 
\sum_{j: v_i v_j\in E} \sum_{p\in T_j} p 
\right)
\right\rangle
\\
&=\dotp{p_s}{U_i}, 
\end{align*}
where $U_i=2\left( \|q_i\|^2 \sum_{p\in T_i} p - 
\sum_{j: v_i v_j\in E} \sum_{p\in T_j} p 
\right)$
and $T_j$ is the set of points in $p_{s+1},\dots,p_n$ that was lifted
using $q_j$ in the traversal.
Collecting the three terms, we get 
\begin{equation}
\|p_s\|^2 \|q_i\|^2 + \dotp{p_s}{c_{s-1}} R_i + \dotp{p_s}{U_i} 
=\alpha_s N_i + \beta_s R_i + \dotp{p_s}{U_i},
\label{equation:objective}
\end{equation}
with 
\[
N_i=\|q_i\|^2, \alpha_s:=\|p_s\|^2, 
\beta_s:=\dotp{p_s}{c_{s-1}}.
\]
The terms $\alpha_s,\beta_s,p_s$ are fixed for iteration $s$.

\paragraph{Algorithm.}
For each $s\in[1,n]$, we pre-compute the following:
\begin{itemize}
\item prefix sums $\sum_{a=1}^{s} p_{a}$, and

\item $\alpha_s$ and $\beta_s$.
\end{itemize}
With this information, it is straightforward to compute a traversal 
in $O(ndk)$ time by evaluating the expression for each choice of $p_s$.
We describe a more careful method that reduces this time
to $O(nd\ceil{\log  k})$.

We assume that $\qgraph$ is a balanced $\mu$-ary tree.
Recall that each node $v_i$ of $\qgraph$ corresponds to a vector $q_i$.
We augment $\qgraph$ with the following 
additional information for each node $v_i$:
\begin{itemize}
\item $N_i=\|q_i\|^2$: recall that this is the degree of $v_i$.

\item $N^{st}_i$: this is the average of the $N_j$ over all elements $v_j$
in the subtree rooted at $v_i$.
Since the subtree contains both internal nodes and leaves, this value is not $\mu+1$.

\item $r'_i$: as before, this is the number of elements of the set $T_i$ 
of the partition that are yet to be determined.
We initialize each $r'_i:=r_i$. 

\item $R_i=2\left(r_i' N_i-\sum_{v_iv_j\in E}r'_j \right)$, 
that is, $r'_i N_i$ minus the $r'_j$ for each node $v_j$ that is a neighbor 
of $v_i$ in $\qgraph$, times two.
We initialize $R_i:=0$.

\item $R^{st}_i$: this is the average of the $R_j$ values over all 
nodes $v_j$ in the subtree rooted at $v_i$.
We initialize this to $0$.

\item $T_i, u_i$: as before, $T_i$ is the set of vectors of the traversal 
that was lifted using $q_i$. 
The sum of the vectors of $T_i$ is $u_i$.
We initialize $T_i=\varnothing$ and $u_i=\orig$.

\item $U_i=2\left( \|q_i\|^2 \sum_{p\in T_i} p - 
\sum_{j: v_i v_j\in E} \sum_{p\in T_j} p \right)
=2\left(u_i N_i-\sum_{v_iv_j\in E} u_j \right)$,  initially
$\orig$.

\item $U^{st}_i$: this is the average of the vectors $U_j$ for all 
nodes $v_j$ in the subtree of $v_i$.
$U^{st}$ is initialized as $\orig$ for each node.
\end{itemize}
Additionally, each node contains pointers to its children and parents.
The quantities $N^{st},R^{st}$ are initialized in one pass over $\qgraph$. 

In step $s$, we find an $i\in [k]$ for which 
Equation~\eqref{equation:objective} has a value at
most the average 
\begin{align*}
A_s &=\frac{1}{k}\left(\sum_{i=1}^{k} \alpha_s N_i + \beta_s R_i + \dotp{p_s}{U_i}
\right)\\
&=\frac{\alpha_s}{k} \sum_{i=1}^{k}N_i 
+  \frac{\beta_s }{k}\sum_{i=1}^{k}R_i
+ \left\langle p_s, \frac{1}{k} \sum_{i=1}^{k}U_i \right\rangle 
\\
&=\alpha_s N^{st}_1 + \beta_s R^{st}_1 + \dotp{p_s}{U^{st}_1},
\end{align*}
where $v_1$ is the root of $\qgraph$.
Then $y_s$ satisfies Equation~\eqref{equation:conditional-ex}.

To find such a node $v_i$, we start at the root $v_1\in \qgraph$.
We compute the average $A_s$
and evaluate Equation~\eqref{equation:objective} at $v_1$.
If the value is at most $A_s$, we report success, setting $i=1$.
If not, then for at least one child $v_m$ of $v_1$,
the average for the subtree is less than $A_s$, that is, 
\[
\alpha_s N^{st}_m + \beta_s R^{st}_m + \dotp{p_s}{U^{st}_m} < A_s.
\]
We scan the children of $v_1$ and compute the expression 
to find such a node $v_m$.
We recursively repeat the procedure
on the subtree rooted at $v_m$, and so on, until we find
a suitable node.
There is at least one node in the subtree at $v_m$ for which 
Equation~\eqref{equation:objective} evaluates to less than $A_s$,
so the procedure is guaranteed to find such a node.

Let $v_i$ be the chosen node.
We update the information stored in the nodes of the tree
for the next iteration.
We set
\begin{itemize}
\item $r'_i:=r'_i-1$ and $R_i:=R_i - 2N_i$.
Similarly we update the $R_i$ values for neighbors of $v_i$.

\item We set $T_i:=T_i\cup \{p_s\}$, $u_i:=u_i+p_s$ and $U_i:=U_i + 2 N_i p_s$.
Similarly we update the $U_i$ values for the neighbors.

\item For each child of $v_i$ and each ancestor of $v_i$ on the
path to $v_1$, we update $R^{st}$ and $U^{st}$.
\end{itemize}
After the last step of the algorithm, we get the required partition $T_1,\dots,T_k$ of $P$.
This completes the description of the algorithm.

\paragraph{Runtime.}
Computing the prefix sums and $\alpha_s,\beta_s$ takes $O(nd)$ time in total.
Creating and initializing the tree takes $O(k)$ time.
In step $s$, computing the average $A_s$ and 
evaluating Equation~\eqref{equation:objective} takes $O(d)$ time per node.
Therefore, computing Equation~\eqref{equation:objective} for the children 
of a node takes $O(d\mu)$ time, as $\qgraph$ is a $\mu$-ary tree.
In the worst case, the search for $v_i$ starts at the root and goes to a leaf,
exploring $O(\mu\ceil{\log_\mu k})$ nodes in the process and 
hence takes $O(d\mu\ceil{\log_\mu k})$ time.
For updating the tree, the information local to $v_i$ and its neighbors
can be updated in $O(d\mu)$ time.
To update $R^{st}$ and $U^{st}$ we travel on the path to the root, which 
can be of length $O(\ceil{\log_\mu k})$ in the worst case, and hence 
takes $O(d\mu\ceil{\log_\mu k})$ time.
There are $n$ steps in the algorithm, 
each taking $O(d\mu \ceil{\log_\mu k})$ time.
Overall, the running time is $O(nd\mu\ceil{\log_\mu k})$ which is minimized 
for a $3$-ary tree.
\qed

\subsubsection{Algorithm for the balanced case}
\label{subsubsection:algo-balanced}

In the case of balanced traversals, 
$\qgraph$ is chosen to be a star graph
as was done in Section~\ref{subsection:tverberg2}.
Let $q_1$ correspond to the root of the graph and $q_2,\dots,q_k$
correspond to the leaves.
In this case the objective function 
$
\alpha_s N_i + \beta_s R_i + \dotp{p_s}{U_i}
$
from the general case
can be simplified:
\begin{itemize}
\item for $i=2,\dots,k$, we have that
$R_i=2\left(r'_i\|q_i\|^2-\sum_{v_iv_j\in E}r'_j\right)
=2\left(r'_i-r'_1\right)$.
Also, we have 
\[
U_i =2\left( \sum_{p\in T_i} p \|q_i\|^2 - 
\sum_{  \substack{p\in T_j \\ v_i v_j\in E} } p \right)
=2\left(\sum_{p\in T_i} p - \sum_{p\in T_1} p \right).
\]

\item for the root $v_1$,
$R_i=2\left(r'_i\|q_i\|^2-\sum_{v_iv_j\in E}r'_j\right)
=2\left((k-1)r'_1-\sum_{j=2}^{k}r'_j\right)$.
Also, we can write 
\[
U_i =2\left( \|q_i\|^2 \sum_{p\in T_i} p - 
\sum_{ \substack{p\in T_j \\ v_i v_j\in E} } p \right)
=2\left((k-1)\sum_{p\in T_i} p - \sum_{p\in \cup_{j=2}^{k} T_j} p \right).
\]
\end{itemize}

We augment $\qgraph$ with information at the nodes
just as in the general case, and use the algorithm to compute the traversal.
However, this would need time $O(nd\mu \ceil{\log_\mu k})=O(ndk)$ 
since $\mu=(k-1)$ and the height of the tree is 1.
Instead, we use an auxiliary balanced ternary rooted tree $\tree$ for the algorithm,
that contains $k$ nodes, each associated to one of the vectors
$q_1,\dots,q_k$ in an arbitrary fashion.
We augment the tree with the same information as in the general case,
but with one difference: for each node $v_i$,
the values of $R_i$ and $U_i$ are updated according to the 
adjacency in $\qgraph$ and not using the edges of $\tree$.
Then we can simply use the algorithm for the general case to get a balanced
partition.
The modification does not affect the complexity of the algorithm.

\section{No-dimensional Colorful Tverberg Theorem}
\label{section:theorem2}

In this section, we prove \Cref{theorem:colorful-tverberg}
and give an algorithm to compute a
colorful partition.
\ndcolorfultverberg*
The general approach is similar to that in Section~\ref{section:theorem1},
but the lifting and the averaging steps are modified.

\subsection{Proof of \Cref{theorem:colorful-tverberg}(\ref{ctverberg1})(colorful partition)}
\label{subsection:ctverberg1}
Let $q_1,\dots,q_k$ be the set of vectors derived from 
a graph $\qgraph$ as in Section~\ref{section:setup}.
Let $\pi=(1,2,\dots,k)$ be a permutation of $[k]$.
Let $\pi_i$ denote the permutation obtained by cyclically shifting 
the elements of $\pi$
to the left by $i - 1$ positions.
That means, 
\begin{align*}
\pi_1&= (1,2,\dots,k-1,k)\\
\pi_2&= (2,3,\dots,k,1)\\
\pi_3&= (3,4,\dots,1,2)\\
&\dots \\
\pi_{k}&= (k,1,2,\dots,k-2,k-1).
\end{align*}
Let $P_1,\dots,P_n$ be point sets in $\R^d$, each of cardinality $k$.
Let $P_1=\{p_{1,1},\dots,p_{1,k} \}$ and
$P_{1,j}=\sum_{i=1}^{k} p_{1,i}\otimes q_{\pi_j(i)}$
be the point in $\R^{d\numedges}$
that is formed by taking tensor products of the points of $P_1$ with 
the permutation $\pi_j$ of $q_1,\dots,q_k$ and adding them up, 
for $j\in [k]$.
For instance, 
$P_{1,4}=p_1\otimes q_4 + p_2\otimes q_5 + \dots + p_k\otimes q_3$.
This gives us a set of $k$ points $P'_1=\{P_{1,1},\dots,P_{1,k} \}$.
Furthermore,
\begin{align}
\sum_{j=1}^{k} P_{1,j} &= \sum_{j=1}^{k} 
\sum_{i=1}^{k} p_{1,i}\otimes q_{\pi_j(i)}
= \sum_{i=1}^{k}  \sum_{j=1}^{k} p_{1,i}\otimes q_{\pi_j(i)}
\nonumber\\
&=\sum_{i=1}^{k}  p_{1,i} \otimes \left( \sum_{j=1}^{k}  q_{\pi_j(i)} \right)
= \sum_{i=1}^{k}  p_{1,i}  \otimes \left( \sum_{m=1}^{k} q_{m} \right) 
\nonumber\\
&=\orig,
\label{equation:colorful-centroid}
\end{align}
so the centroid of $P'_1$ coincides with the origin.
In a similar manner, for $P_2,\dots,P_n$, 
we construct the point sets $P'_2,\dots,P'_n$,
respectively, each of whose centroids coincides with the origin.
We now upper bound $\diam(P'_1)$.
For any point $P_{1,i}$, using Equation~\eqref{equation:tree-product}
we can bound the squared norm as
\begin{align*}
\|P_{1,i}\|^2& 
= \left\| \sum_{m=1}^{k} p_{1,m} \otimes q_{\pi_i(m)} \right\|^2
= \left\| \sum_{l=1}^{k} p_{1,\pi_i^{-1}(l)} \otimes q_{l}  \right\|^2\\
&=\sum_{v_lv_m\in E} \left\|
p_{1,\pi_i^{-1}(l)} - p_{1,\pi_i^{-1}(m)} \right \|^2
\\
&\le \sum_{v_lv_m\in E} \diam(P_1)^2 
\le \numedges \diam(P_1)^2,
\end{align*}
so that 
$\|P_{1,i}\| \le \sqrt{\numedges} \diam(P_1)$. 
For any two points $P_{1,i},P_{1,j}\in P'_1$,
\begin{align*}
\|P_{1,i}-P_{1,j}\| & \le \|P_{1,i}\| + \|P_{1,j}\|
\\
&\le \sqrt{\numedges} \diam(P_1) + \sqrt{\numedges} \diam(P_1)
\\
&= 2\sqrt{\numedges} \diam(P_1).
\end{align*}
Therefore, $\diam(P'_1)\le 2\sqrt{\numedges}\diam(P_1)$.
We get a similar relation for each $P'_{i}$.
Now we apply the no-dimensional Colorful \Caratheodory 
theorem from~\cite[Theorem~2.1]{abm-nodim}
on the sets $P'_1,\dots,P'_n$: there is a traversal
$X=\{ x_1\in P'_1,\dots ,x_n\in P'_n \}$ such that 
\begin{align*}
\|c(X)\| &< \delta=\frac{\mathrm{max}_{i} \diam(P'_i) }{\sqrt{2n}}\\
&\le \frac{2\sqrt{\numedges}}{\sqrt{2n}} \mathrm{max}_{i} \diam(P_i)
= \sqrt{\frac{2k\numedges}{N}} \mathrm{max}_{i} \diam(P_i).
\end{align*}
Let $x_1=P_{1,i_1},\dots,x_n=P_{n,i_n}$ where $1\le i_1,\dots,i_n\le k$
are the indices of the permutations of $\pi$ that were used.
That means, 
\[
x_j=P_{j,i_j}=\sum_{l=1}^{k} p_{j,l}\otimes q_{\pi_{i_j}(l)}
=\sum_{m=1}^{k} p_{j,\pi_{i_j}^{-1}(m)}\otimes q_{m}.
\]
Then, we define the colorful sets $A_1,\dots,A_k$ as:
\[
A_j:=\left\{p_{1,\pi_{i_1}^{-1}(i)}, p_{2,\pi_{i_2}^{-1}(i)}, \dots 
p_{n,\pi_{i_n}^{-1}(i)} \right\},
\]
that is, $A_j$ consists of the points of $P_1,\dots,P_n$ that 
were lifted using $q_j$ for $j \in [k]$.
By definition, each $A_j$ contains precisely one point from each $P'_i$,
so it is a colorful set.
Let $c_j$ denote the centroid of $A_j$.
We expand the expression 
\begin{align*}
c(X)&=\frac{1}{n}\sum_{j=1}^{n} P_{j,i_j} 
\\
&= \frac{1}{n}\sum_{j=1}^{n} \sum_{l=1}^{k} p_{j,l} \otimes q_{\pi_{i_j} (l)}
\\
&= \frac{1}{n}\sum_{j=1}^{n} \sum_{m=1}^{k} p_{j,\pi_{i_j}^{-1} (m)} 
\otimes q_{m}
\\
&=\frac{1}{n}\sum_{m=1}^{k} \sum_{j=1}^{n} p_{j,\pi_{i_j}^{-1} (m)} 
\otimes q_{m}
\\
&= \frac{1}{n}\sum_{m=1}^{k} \left(\sum_{j=1}^{n} p_{j,\pi_{i_j}^{-1} (m)} \right)\otimes q_{m}
\\
&= \sum_{m=1}^{k} \frac{1}{n}\left(\sum_{j=1}^{n} p_{j,\pi_{i_j}^{-1} (m)} \right)\otimes q_{m}
\\
&=\sum_{m=1}^{k} c_m \otimes q_m. 
\end{align*}
Applying $\|c(X)\|^2<\delta^2 $, we get 
\begin{align*}
\left\| \sum_{m=1}^{k} c_m \otimes q_m \right\|^2 
=\sum_{v_l,v_m \in E} \|c_l-c_m\|^2 < \delta^2,
\end{align*}
where we made use of Equation~\eqref{equation:tree-product}.
Using the Cauchy-Schwarz inequality as in Section~\ref{subsection:tverberg1},
the distance from $c_1$ to any other $c_j$ is at most $\sqrt{\gdiam}\delta$.
Substituting the value of $\delta$, this is
$\sqrt{\frac{2k\gdiam\numedges}{N}} \mathrm{max}_{i} \diam(P_i)$.
Now we set $\qgraph$ as a star graph, similar to the balanced case
of Section~\ref{subsection:tverberg2} with $v_1$ as the root.
A ball of radius 
\[
\sqrt{\frac{2k(k-1)}{N}}  \mathrm{max}_{i} \diam(P_i)
\]
centered at $c_1$ contains the set $\{ c_1,\dots,c_k \}$, 
intersecting the convex hull of each $A_j$, as required.
\qed

\subsection{Proof of \Cref{theorem:colorful-tverberg}(\ref{ctverberg2})(computing the colorful partition)}
\label{subsection:ctverberg2}
The algorithm follows a similar approach as in Section~\ref{subsection:tverberg3}.
The input consists of the sets of points $P_1,\dots,P_n$. 
We use the permutations
$\pi_1,\dots,\pi_k$ of $q_1,\dots,q_k$ to (implicitly) construct 
the point sets $P'_1,\dots,P'_n$.
Then we compute a traversal of $P'_1,\dots,P'_n$ using the method
of conditional expectations.
This essentially means determining a permutation $\pi_{i_j}$
for each $P'_i$.
The permutations directly determine the colorful partition.
Once again, we do not explicitly lift any vector
using the tensor product, and thereby avoid the associated costs.

We iterate over the points of $\{P'_1,\dots,P'_n \}$ in reverse order
and find a suitable
traversal $Y=(y_1\in P'_1,\dots, y_n\in P'_n)$ point by point.
Suppose we have already selected the points $\{ y_{s+1},y_{s+2},\dots,y_n \}$.
To find $y_s\in P'_s$, it suffices to choose any point that satisfies
\begin{equation}
\av\left(\|c(x_1,\dots,x_{s-1},y_{s},y_{s+1},\dots,y_n)\|^2\right)\le 
\av\left(\|c(x_1,\dots,x_{s},y_{s+1},\dots,y_n)\|^2\right).
\end{equation}
Specifically, we find the point $y_s$ for which 
the conditional expectation expressed as 
\[
\av(\|c(x_1,x_2,\dots,x_{s-1},y_{s},\dots,y_n)\|^2) 
\]
is minimized.
As in Equation~\eqref{equation:expectedsum} from Section~\ref{subsection:tverberg3}, 
this is equivalent to determining the point that minimizes
\begin{align}
&\|y_s\|^2 + 2\left\langle y_s,
\av\left(\sum_{i=1}^{s-1}x_i\right) + \sum_{i=s+1}^{n}y_i 
\right\rangle
\\
=&
\|y_s\|^2 + 2\left\langle y_s,
\av\left(\sum_{i=1}^{s-1}x_i\right)\right\rangle 
+ 2\dotp{y_s}{\sum_{i=s+1}^{n}y_i}. 
\label{equation:expectedsum-colorful}
\end{align}

Let $y_s=\sum_{i=1}^{k} p_{s,i} \otimes q_{\pi(i)}$
for some permutation $\pi\in\{\pi_1,\dots,\pi_k  \}$.
The terms of Equation~\eqref{equation:expectedsum-colorful}
can be expanded as:
\begin{itemize}
\item first term:
\[
\|y_s\|^2 = \left\| \sum_{i=1}^{k} p_{s,i} \otimes q_{\pi(i)} \right\|^2
=\left\| \sum_{l=1}^{k} p_{s,\pi^{-1}(l)} \otimes q_{l} \right\|^2
= \sum_{v_lv_m \in E} \left\|  p_{s,\pi^{-1}(l)} -p_{s,\pi^{-1}(m)}
\right\|^2,
\]
using Equation~\eqref{equation:tree-product}.

\item second term: the expectation can be written as 
\begin{align*}
\av\left(\sum_{i=1}^{s-1}x_i\right)&
= \sum_{i=1}^{s-1} \sum_{j=1}^{k} P_{i,j}\frac{1}{k}
=\frac{1}{k}\sum_{i=1}^{s-1} \left( \sum_{j=1}^{k} P_{i,j} \right)
=\orig,
\end{align*}
as in Equation~\eqref{equation:colorful-centroid}.

\item third term: 
let $\pi_{j_{s+1}},\dots,\pi_{j_{n}}$ denote the respective permutations selected
for $P'_{s+1},\dots,P'_{n}$ in the traversal.
Then, 
\begin{align*}
\sum_{i=s+1}^{n}y_i &=\sum_{i=s+1}^{n} P_{i,j_{i}}
\\
&=\sum_{i=s+1}^{n} \sum_{l=1}^{k} p_{i,l}\otimes q_{\pi_{j_{i}}(l)}
\\
&=\sum_{i=s+1}^{n} \sum_{m=1}^{k} p_{i,\pi^{-1}_{j_{i}}(m)}\otimes q_{m}
\\
&=\sum_{m=1}^{k} \left( \sum_{i=s+1}^{n}  p_{i,\pi^{-1}_{j_{i}}(m)} \right)
\otimes q_{m}
\\
&= \sum_{m=1}^{k} \sum_{p\in A'_{m}}p \otimes q_{m},
\end{align*}
where, $A'_m\subseteq A_m$ is the colorful set whose elements 
from $P_{s+1},\dots,P_n$ have already been determined.
Let $S_m= \sum_{p\in A'_{m}}p$ for each $m=1\dots k$.
Then, the third term can be written as
\begin{align*}
2 \left\langle y_s,\sum_{i=s+1}^{n}y_i \right\rangle&=
2\left\langle \sum_{i=1}^{k} p_{s,i} \otimes q_{\pi(i)},
\sum_{m=1}^{k} S_{m} \otimes q_{m} \right\rangle \\
&= 2 \sum_{i=1}^{k}  \sum_{m=1}^{k} \left\langle  
p_{s,i} \otimes q_{\pi(i)}, S_{m} \otimes q_{m} \right\rangle \\
&=2 \sum_{l=1}^{k}  \sum_{m=1}^{k} \left\langle  
p_{s,\pi^{-1}(l)} \otimes q_{l}, S_{m} \otimes q_{m} \right\rangle
\\
&=2 \sum_{l=1}^{k}  \sum_{m=1}^{k} \left\langle  p_{s,\pi^{-1}(l)}, S_{m}
\right\rangle \dotp{q_{l}}{q_{m}} \\
&=2\sum_{m=1}^{k} \left(
\left\langle  p_{s,\pi^{-1}(m)}, S_{m} \right\rangle \|q_m\|^2-
\sum_{v_lv_m\in E} \left\langle  p_{s,\pi^{-1}(l)}, S_{m}\right\rangle
\right)\\
&=2\sum_{m=1}^{k} \left\langle
\left( p_{s,\pi^{-1}(m)}\|q_m\|^2 - \sum_{v_lv_m\in E} 
p_{s,\pi^{-1}(l)} \right),
S_{m}\right\rangle.
\end{align*}
\end{itemize}
If $\tau$ is the permutation selected in the iteration for $P'_s$,
then we update $A'_i=A'_{i}\cup \{p_{s,\tau^{-1}(i)}\}$ 
and $S_i=S_i + p_{s,\tau^{-1}(i)}$ for each $i=1,\dots,k$.

For each permutation $\pi$, the first and the third terms 
can be computed in $O(\numedges d)=O(kd)$ time.
There are $k$ permutations for each iteration, so this takes $O(k^2d)$ time
per iteration and $O(nk^2d)=O(Ndk)$ time in total for finding
the traversal.

\begin{remark}
In principle, it is possible to reduce the problem of computing a
no-dimensional Tverberg partition to the problem of computing a
no-dimensional Colorful Tverberg partition.
This can be done by arbitrarily coloring the point set into 
sets of equal size, and then using the algorithm for the colorful version.
This can give a better upper bound on the radius of the 
intersecting ball if the diameters of the colorful sets satisfy
\[
\mathrm{max}_i \diam(P_i) 
< \frac{\diam(P_1\cup P_2\cup \dots \cup P_n)}{\sqrt{2}}.
\]
However, the algorithm for the colorful version has a worse runtime
since it does not utilize the optimizations used in the regular version.
\end{remark}

\section{No-dimensional Generalized Ham-Sandwich Theorem}
\label{section:theorem3}

We prove \Cref{theorem:nodim-gghs} in this section:

\ndinterpolate*

This is a no-dimensional version of 
a generalization of the Ham-Sandwich theorem~\cite{st-sandwich}.
We briefly describe the history of the problem before
detailing the proof.

The Centerpoint theorem was proven by Rado in~\cite{rado-centerpoint}.
It states that for any set of $n$ points $P\subset \R^d$, 
there exists some point $\cp(P)\in \R^d$, called the \emph{centerpoint}
of $P$, such that $\cp(P)$ has depth
at least $\left\lceil n/(d+1)\right\rceil$.
The centerpoint generalizes the concept of median to 
higher dimensions.
The theorem can be proven using
Helly's theorem~\cite{helly} or Tverberg theorem.

The Ham-Sandwich theorem~\cite{st-sandwich} shows that
for any set of $d$ finite point sets $P_1,\dots,P_d\subset \R^d$,
there is a hyperplane $H$ which bisects each point set, 
that is, each closed halfspace defined by $H$ contains at least
$\left\lceil |P_i|/2 \right\rceil$ points of $P_i$,
for $i \in [d]$.
The result follows by an application of the Borsuk-Ulam
theorem~\cite{mt-bubook}.

Zivaljevi\'{c} and Vre\'{c}ica~\cite{zv-extension} and 
Dol'nikov~\cite{dolnikov-generalization}, independently, 
proved a generalization of these two results for affine subspaces (\textit{flats})
:

\begin{theorem}
\label{theorem:centerpoint-hs}
Let $P_1,\dots,P_k$ be $k\le d$ finite point sets in $\R^d$.
Then there is a $(k-1)$-dimensional flat $F$ of depth at least
$|P_i|/(d-k+2)$ with respect to $P_i$, for $i \in [k]$.
\end{theorem}
For $k=1$, this corresponds to the Centerpoint theorem while 
for $k=d$, this is the Ham-Sandwich theorem, and thereby interpolates
between the two extremes.

We prove a no-dimensional version of this theorem, where 
$1/(d-k+2)$ can be relaxed to be an arbitrary 
but reasonable fraction.
In fact, we prove a slightly stronger version that allows an 
independent choice of fraction for each point set $P_i$ individually.
The idea is motivated by the result of 
\Barany, Hubard and Jer\'onimo, who showed in~\cite{bhj-ghs} that 
under certain conditions
of ``well-separation", $d$ compact sets $S_1,\dots,S_d\subset\R^d$ 
can be divided by a hyperplane that such the positive half-space 
contains an $(\alpha_1,\dots,\alpha_d)$-fraction of the volumes 
of $S_1,\dots,S_d$, respectively.
A discrete version of this result for finite point sets
was proven by Steiger and Zhao
in~\cite{sz-ghs}, which they term as the 
\emph{Generalized Ham-Sandwich theorem}.
Our result can be interpreted as a no-dimensional version of this result,
but we do not have constraints on the point sets
as in~\cite{bhj-ghs,sz-ghs}.

Without loss of generality, we assume that the centroid $c(P_1)=\orig$.
We first approach a simpler case:

\begin{lemma}
\label{lemma:easy-centerpoint}
Let $c(P_1)=\dots=c(P_k)=\orig$
and $m_1, \dots, m_k$, $2 \leq m_i \leq |P_i|$ for $i \in [k]$,
be any choice of integers.
Then the ball of radius 
\[
(2+2\sqrt{2})\max_i \frac{\diam(P_i)}{\sqrt{m_i}}
\]
centered at $\orig$ has depth at least 
$\left\lceil |P_i|/m_i \right\rceil$
with respect to $P_i$, for $i \in [k]$.
\end{lemma}

\begin{proof}
Consider any point set $P_i$
and a no-dimensional $\left\lceil \frac{|P_i|}{m_i} \right\rceil$-partition
of $P_i$.
From~\cite[Theorem\,2.5]{abm-nodim}, we know that the ball $B$
centered at $c(P_i)=\orig$ of radius 
\begin{equation*}
(2+\sqrt{2})\diam(P_i)
\sqrt{\frac{\left\lceil |P_i|/m_i \right\rceil}{|P_i|}}
<(2+\sqrt{2})\diam(P_i) \sqrt{\frac{2}{m_i}}
=\frac{(2+2\sqrt{2})\diam(P_i)}{\sqrt{m_i}}
\end{equation*}
intersects each set of the partition.
Let $H$ be any half-space that contains $B$.
We claim that $H$ contains at least one point from each set in 
the partition.
Assume for contradiction that $H$ does not contain
any point from a given set in the partition. 
Then, the convex hull of that set does not intersect $H$, and hence
$B$, which is a contradiction.
This shows that $B$ has depth $\left\lceil |P_i|/m_i \right\rceil$.
Let $B'$ be the ball of radius 
$(2+2\sqrt{2})\max_i \diam(P_i)/\sqrt{m_i}$
centered at the origin.
Then $B'$ has depth at least $\left\lceil |P_i|/m_i \right\rceil$
with respect to $P_i$ for each $i=1,\dots,k$.
\end{proof}
We prove an auxiliary result that will be helpful in proving
the main result:
\begin{lemma}
\label{lemma:projection}
Let $P_1,\dots,P_k \subset \R^{d_1}$ be finite point sets.
Let $v$ be any vector in $\R^{d_1}$ and 
project $P_1,\dots,P_k$ on the hyperplane $H$ via $\orig$ with normal $v$.
If some set $X\subset H$ has depth $\alpha_1,\dots,\alpha_d$ respectively 
for the projected point sets,
then $X\times \R_v \subset \R^{d_1}$ has the same depths for the original
point sets, where $\R_v$ is the one dimensional subspace containing $v$.
\end{lemma}

\begin{proof}
Consider any half-space $\mathcal{H}\subset \R^{d_1}$ that contains
$X\times \R_v$.
Then $\mathcal{H}$ contains $\R_v$, 
so it can be written as $\hat{\mathcal{H}}\times \R_v$, where $\hat{\mathcal{H}}\subset H$
is a half-space containing $X$.
$\hat{\mathcal{H}}$ contains at least $\alpha_i$ points of each $P_i$.
By orthogonality of the projection, $\mathcal{H}$ also contains at least 
$\alpha_i$ points of each $P_i$, proving the claim.
\end{proof}
\paragraph{Proof of \Cref{theorem:nodim-gghs}(\ref{gghs1}).}

Given point sets $P_1,\dots,P_k$ with $c(P_1)=\orig$,
we apply orthogonal projections on the points multiple times
so that their centroids coincide.
In the first step, we set $v_1=c(P_2)$.
Let $l_1$ be the line through the origin containing $v_1$ and let
$H_{v_1}$ be the hyperplane via $\orig$ with normal $v_1$.
Let $f_1:\R^d \rightarrow H_{v_1}$ be the orthogonal projection
defined as $f(p)=p-\dotp{p}{v}\frac{v}{|v|^2}$.
Let $P_1^{1},\dots,P_k^{1}\subset \R^{d-1}$ be the point sets 
obtained by applying the
orthogonal projection on $P_1,\dots,P_{k}$, respectively.
Under this projection $c(P_1^{1})=c(P_2^{1})=\orig$.
In the next step we set $v_2=c(P_3^{1})$ and define $l_2$ and $H_{v_2}$
analogously.
We project $P_1^{1},\dots,P_k^{1}$ onto $H_{v_2}$ to get 
$P_1^{2},\dots,P_k^{2}$ with $c(P_1^{2})=c(P_2^{2})=c(P_3^{2})=\orig$.
We repeat this process $k-1$ times to get
a set of points $P_1^{k-1},\dots,P_k^{k-1}\subset \R^{d-k+1}$ with 
$c(P_1^{k-1})=\dots=c(P_k^{k-1})=\orig$.
Using Lemma~\ref{lemma:easy-centerpoint}, there is a ball $B$
of radius 
\[
(2+2\sqrt{2})\max_i \frac{\diam(P_i^{k-1})}{\sqrt{m_i}}
<(2+2\sqrt{2})\max_i \frac{\diam(P_i)}{\sqrt{m_i}}
\]
of the required depth.
Applying Lemma~\ref{lemma:projection}
on $P_1^{k-2},\dots,P_k^{k-2}\subset \R^{d-k+2}$,
$B\times \ell_{k-1}$ also has the required depth. 
Repeated application of Lemma~\ref{lemma:projection}
gives us $B\times \ell_{k-1}\times \ell_{k-2}\times\dots\times \ell_{1}$.
Since the Cartesian product may have more than $d$ co-ordinates,
we apply a linear transformation so that the subspace spanned by 
the orthogonal set $\ell_1,\dots,\ell_{k-1}$
is $\R^{k-1}$.
Then, $B\times \R^{k-1}$ has the desired properties.

\paragraph{Proof of \Cref{theorem:nodim-gghs}(\ref{gghs2}).}

To compute the vectors $v_1,\dots,v_{k-1}$, we note that 
\[
v_i=c(P_{i+1}^{i-1})=c(f_{i-1}\circ f_{i-2}\circ \dots \circ f_{1}(P_{i+1}^{i-1}))
=f_{i-1}\circ f_{i-2}\circ \dots \circ f_{1} (c(P_{i+1}^{i-1})),
\]
by linearity of the projection.
Therefore, at the beginning we first compute each centroid $c(P_i)$
and in each step we apply the projection on the relevant centroids.
The projection is applied $1+\dots+k-2=O(k^2)$ times.
Computing the centroid in the first step takes $O(\sum_{i} |P_i|d)$ time.
Computing the projection once takes $O(d)$ time, so in total
$O(dk^2)$ time.
Finding the linear transformation takes another $O(d^6)$ time.

\section{Conclusion and future work}
\label{section:conclusion}

We gave efficient algorithms for a no-dimensional version of Tverberg
theorem and for a colorful counterpart.
To achieve this end, we presented a refinement of Sarkaria's 
tensor product construction by defining vectors using a graph.
The choice of the graph was different for the general- and the 
balanced-partition cases and also influenced the time complexity 
of the algorithms.
It would be interesting to find more applications
of this refined tensor product method.
Another option could be to look at non-geometric generalizations
based on similar ideas. It would also be interesting to 
consider no-dimensional variants other generalizations of Tverberg's
theorem, e.g., in the tolerant setting~\cite{MulzerSt14,soberon-equal}.

The radius bound that we obtain for the Tverberg partition
is $\sqrt{k}$ off the optimal bound in~\cite{abm-nodim}.
This seems to be a limitation in handling
Equation~\eqref{equation:balanced-centroid-norm}.
It is not clear if this is an artifact of using 
tensor product constructions.
It would be interesting to explore if this factor can be brought down
without compromising on the algorithmic complexity.
In the general partition case, setting $r_1=\dots=r_k$ gives a bound
that is $\sqrt{\ceil{\log k}}$ worse than the balanced case,
so there is some scope for optimization.
In the colorful case, the radius bound is again $\sqrt{k}$ 
off the optimal~\cite{abm-nodim}, but with a silver lining.
The bound is proportional to $\max_i{\diam(P_i)}$ in contrast
to $\diam(P_1\cup \dots \cup P_n)$ in~\cite{abm-nodim}, which is
better when the colors are well-separated.

The algorithm for colorful Tverberg theorem has a worse runtime than
the regular case.
The challenge in improving the runtime lies a bit with selecting 
an optimal graph as well as the nature of the problem itself.
Each iteration in the algorithm looks at each of the permutations
$\pi_1,\dots,\pi_k$ and computes the respective expectations.
The two non-zero terms in the expectation are both computed
using the chosen permutation.
The permutation that minimizes the first term can be determined 
quickly if $\qgraph$ is chosen as a path graph.
This worsens the radius bound by $\sqrt{k-1}$.
Further, computing the other (third) term of the expectation still requires 
$O(k)$ updates per permutation and therefore $O(k^2)$ updates
per iteration, thereby eliminating
the utility of using an auxiliary tree to determine the best
permutation quickly.
The optimal approach for this problem is unclear at the moment. 

\bibliographystyle{plain}
\bibliography{references}

\end{document}